\documentclass[10pt,letterpaper]{article}

\usepackage{fullpage}
\usepackage{textcomp}

\usepackage{amsmath, graphicx, epsfig, color, amsfonts, algorithm, amssymb, mathtools, enumitem, amsthm}
\usepackage{color, soul, xcolor}

\usepackage[font=footnotesize]{caption}
\usepackage{subcaption}

\makeatletter
\let\NAT@parse\undefined
\makeatother

\usepackage[round, sort&compress, numbers]{natbib}
\usepackage{float}

\graphicspath{{./figures_arxiv/}}

\usepackage{version,xspace}
\usepackage[noend]{algorithmic}

\def\expt{\mathbb{E}}
\def\real{\mathbb{R}}

\newcommand{\subscr}[2]{#1_{\textup{#2}}}
\newcommand{\supscr}[2]{#1^{\textup{#2}}}

\newcommand\oprocendsymbol{\hbox{$\square$}}
\newcommand\oprocend{\relax\ifmmode\else\unskip\hfill\fi\oprocendsymbol}

\def \bs {\boldsymbol}
\def \mc {\mathcal}
\def \mb {\mathbb}
\def\tran{^\top}

\newtheorem{theorem}{Theorem}

\newtheorem{lemma}{Lemma}

\newtheorem{remark}{Remark}

\usepackage{array}

\usepackage{soul}

\newcommand{\norm}[1]{\left\lVert#1\right\rVert}

\title{\LARGE \bf
Co-Learning Port-Hamiltonian Systems and Optimal Energy-Shaping Control\thanks{This work has been supported by NSF Award CMMI-1940950}
}

\author{Ankur Kamboj$^{1}$, Biswadip Dey$^{2}$, and Vaibhav Srivastava$^{1}$
\thanks{$^1$Electrical and Computer Engineering, Michigan State University, East Lansing,
MI 48824. \texttt{\{ankurank, vaibhav\}@msu.edu}.}
\thanks{$^2$
Meta Reality Labs, Redmond, WA 98052. \texttt{biswa-dey@ieee.org}}
}

\begin{document}

\date{}
\maketitle
\pagestyle{plain}

%%%%%%%%%%%%%%%%%%%%%%%%%%%%%%%%%%%%%%%%%%%%%%%%%%%%%%%%%%%%%%%%%%%%%%%%%%%%%%%%
\begin{abstract}
We develop a physics-informed learning framework for energy-shaping control of port-Hamiltonian (pH) systems from trajectory data. The proposed approach co-learns a pH system model and an optimal energy-balancing passivity-based controller (EB-PBC) through alternating optimization with policy-aware data collection. At each iteration, the system model is refined using trajectory data collected under the current control policy, and the controller is re-optimized on the updated model. Both components are parameterized by neural networks that embed the pH dynamics and EB-PBC structure, ensuring interpretability in terms of energy interactions. The learned controller renders the closed-loop system inherently passive and provably stable, and exploits passive plant dynamics without canceling the natural potential. A dissipation regularization enforces strict energy decay during training, thereby enhancing robustness to sim-to-real gaps. The proposed framework is validated on state-regulation and swing-up tasks for planar and torsional pendulum systems.
\end{abstract}

% \begin{keyword}
% Passivity-based control, port-Hamiltonian systems, Neural networks, Data-driven control, Stability of nonlinear systems
% \end{keyword}

% \begin{keyword}
% Passivity-based control, port-Hamiltonian systems, Lyapunov methods, Stability of nonlinear systems, Optimization, Neural networks, Data-driven control
% \end{keyword}

% \end{frontmatter}

%%%%%%%%%%%%%%%%%%%%%%%%%%%%%%%%%%%%%%%%%%%%%%%%%%%%%%%%%%%%%%%%%%%%%%%%%%%%%%%%
\section{Introduction}
% Energy-based control has seen a recent surge in interest across several fields, including robotics~\citep{holm2008kinetic, lin2019contact,franco2021energy,chang2023energy}, aerospace~\citep{aoues2019modeling}, and energy systems~\citep{strehle2018towards}. Systematically modeling physical systems as subsystems exchanging energy through ports gave rise to port-Hamiltonian~(pH) theory~\citep{van2014port}, which has been invaluable in modeling complex multi-domain physical systems owing to its modularity. Passivity-based controllers~(PBC) designed for pH systems have the advantage of rendering the closed-loop system passive and stable with respect to a storage function while achieving the control objective~\citep{ortega2002putting}. Moreover, controller design and parameter tuning are intuitive and can be expressed in terms of canceling and adding various energies and dampings in the system~\citep{schaft1999l2}. However, these design techniques~\citep{ortega2008control} are founded on achieving passivity and stability of the closed-loop system, and often ignore any performance considerations. 

Energy-based control has gained significant traction across several disciplines, including robotics, aerospace, and energy systems. 
% ~\citep{holm2008kinetic, lin2019contact,franco2021energy,chang2023energy}, aerospace~\citep{aoues2019modeling}, and energy systems~\citep{strehle2018towards}. 
Port-Hamiltonian (pH) theory~\cite{van2014port} provides a principled framework for modeling such systems as energy-exchanging subsystems, offering modularity and physical interpretability. Passivity-based control (PBC) leverages this structure to ensure closed-loop passivity and stability with respect to a storage function~\citep{ortega2002putting}, with intuitive design via energy shaping and damping injection~\citep{schaft1999l2}. However, classical PBC design primarily targets stability and often does not explicitly address performance optimization~\citep{ortega2008control}.

{Designing high-performance controllers for nonlinear systems remains fundamentally challenging due to model uncertainty, complex dynamics, and computational limitations. Classical model-based approaches offer strong guarantees when accurate models are available, but their performance degrades under mismatch and unmodeled effects. At the same time, purely data-driven methods have emerged as a promising alternative, leveraging trajectory data to synthesize control policies without explicit models. However, generic learning-based approaches often suffer from limited interpretability, weak stability guarantees, unsafe exploration, and slow convergence, limiting their applicability in safety-critical settings. To address these limitations, we propose a physics-informed learning framework that jointly learns structured system models and energy-shaping control policies, combining data-driven adaptability with passivity-based stability guarantees.}

Recent works have explored optimal control within the class of PBCs. Reinforcement learning-based approaches such as Energy-Balancing Actor-Critic (EB-AC)~\citep{sprangers2015reinforcement} and its robust extensions~\citep{gheibi2020designing} incorporate passivity structure to ensure stability and show accelerated convergence when plant model information is incorporated. Neural approaches such as~\citep{massaroli2022optimal, plaza2022total} parameterize energy-shaping policies, while~\citep{okura2020bayesian} and~\citep{sebastian2025physics} extend these ideas to Bayesian and multi-agent settings. However, these methods either assume known dynamics, operate in model-free settings, or impose restrictive structures (e.g., quadratic potentials and linear damping), limiting the ability to exploit intrinsic passive dynamics. Moreover, they often suffer from the dissipation obstacle~\citep{ZHANG2015dissipation}, that is, the desired Hamiltonian cannot be freely chosen.

In parallel, physics-informed learning methods have emerged for learning structured dynamical models. Hamiltonian neural networks~\citep{greydanus2019hamiltonian}, SymODEN~\citep{Zhong2020Symplectic}, and their extensions~\citep{duong2024port,dipersio2025stochastic,van2025learning,beckers2023learning} incorporate physical structure into learned dynamics. Bayesian approaches such as Gaussian Process port-Hamiltonian systems~(GP-PHS)~\citep{beckers2022gaussian} provide uncertainty quantification, \cite{roth2025stable} enforces global Lyapunov stability via convexity constraints, while~\cite{xu2025learning} learn a neural Koopman operator model that learns the dissipativity properties of the true system. 

However, these works primarily focus on system identification and do not jointly learn control policies~\citep{sivaranjani2025control}. More importantly, a critical question that remains largely unaddressed in the literature is: \emph{when a passivity-based controller is trained on an approximate (learned) model, can it provably stabilize the true system despite the sim-to-real gap?} Existing neural Lyapunov approaches~\citep{dai2021lyapunov, chang2019neural} typically assume access to true dynamics during verification.
% , while Bayesian methods offer only probabilistic guarantees.
To address these issues, we propose a physics-informed learning framework that jointly learns pH system models and optimal energy-balancing passivity-based control (EB-PBC) policies from trajectory data via alternating optimization. The model is iteratively refined using a combination of exploratory and policy-driven trajectories, while the controller is re-optimized on the updated model. This policy-aware data collection focuses on model accuracy in the closed-loop operating region. Both the system and controller are parameterized using neural networks that embed pH dynamics and EB-PBC structure, ensuring interpretability and passivity at every training iteration. We further show that the resulting controller is provably stable and robust to model approximation errors.

The proposed approach offers several advantages. First, the learned controller is inherently passive and stabilizes the system at every point during learning, unlike generic RL algorithms that offer no stability guarantees during exploration~\citep{banerjee2025survey}. Second, the NN-parameterization allows the controller to exploit existing passive dynamics rather than canceling them, a shortcoming identified in~\cite{ortega2008control}. Third, the alternating learning scheme improves model fidelity in the policy-relevant region, tightening stability guarantees. Finally, a dissipation regularization enforces strict energy decay during training, improving robustness to sim-to-real gaps. 

The major contributions of this work are fourfold. First, we develop an alternating optimization framework for joint pH system identification and EB-PBC synthesis from trajectory data. Second, we introduce a physics-informed neural architecture that enforces structural passivity and enables expressive energy shaping. Third, we establish deterministic stability guarantees for the learned controller on the true system, with explicit bounds on the convergence region in terms of model error. Fourth, we propose a dissipation regularization mechanism that improves robustness to sim-to-real mismatch.

The rest of the paper is organized as follows. Section~\ref{sec:background} reviews background on pH systems and EB-PBC. Section~\ref{sec:proposedPIDPG} presents the proposed framework. Section~\ref{sec:analysis} provides stability and robustness analysis. Section~\ref{sec:Experiments} presents simulation results, followed by conclusions in Section~\ref{sec:conclusion}.

\begin{figure*}[ht!]
    \centering
    \includegraphics[width=0.75\linewidth]{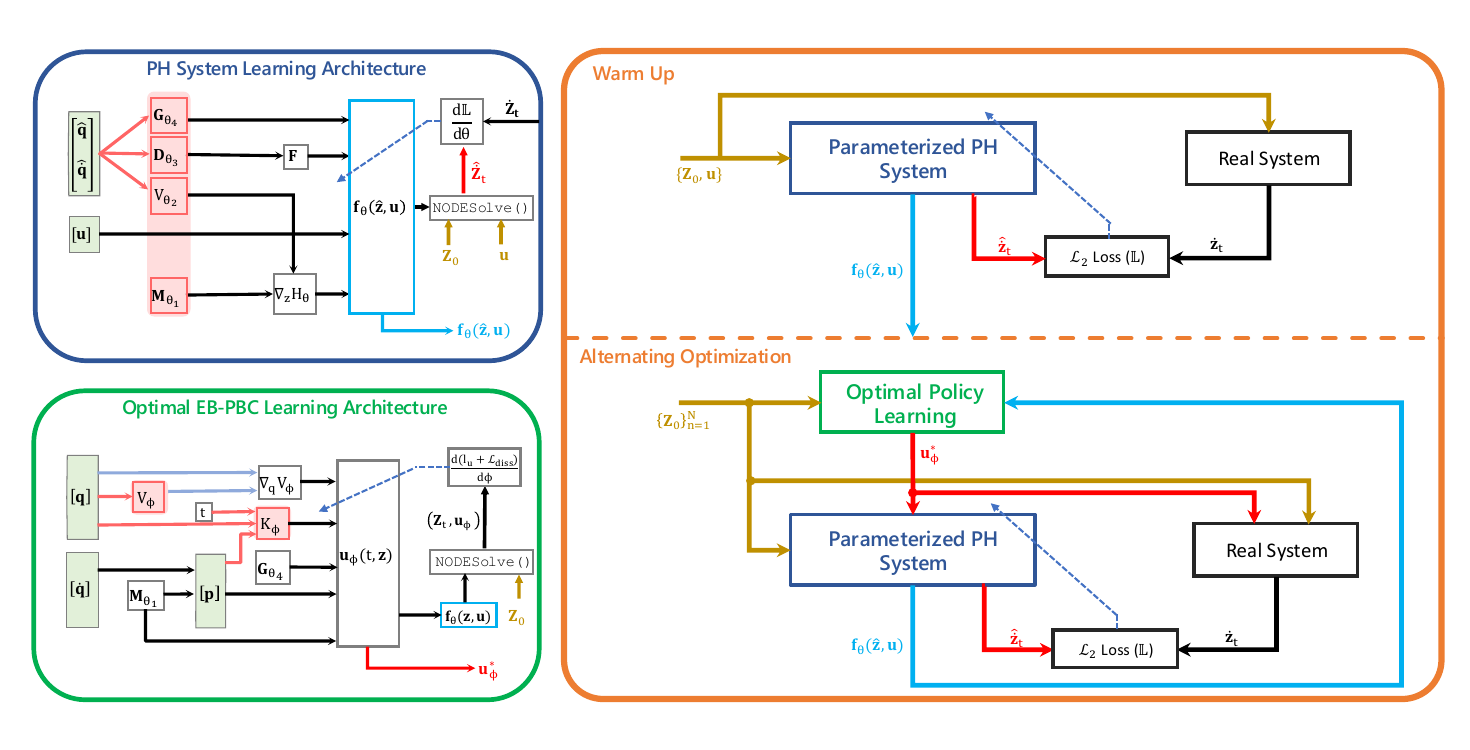}
    \caption{\textbf{Left:} The complete computational graph with different subsystems. Pink blocks and arrows represent NN parameterization. Green blocks are various system states. Purple arrows indicate automatic differentiation to obtain gradients. \textbf{Right:} The training consists of a warm-up phase for initial pH system learning from step-excited data, followed by alternating optimization iterations that refine the system model using policy-aware trajectory data and re-optimize the EB-PBC on the updated model. Different arrow and block colors are consistent with computational graphs.}
    \label{fig:flowchart}
\end{figure*}
\section{Background} \label{sec:background}
In this section, we briefly review Hamiltonian system representations, learning of Hamiltonian dynamics from data, and optimal energy-balancing passivity-based control (EB-PBC).

\subsection{Hamiltonian System Dynamics} \label{sec-phSystem}
Hamiltonian dynamics are a reformulation of Newtonian mechanics, where the dynamics of physical systems are derived from the Hamiltonian function, a scalar quantity representing the total energy of a system. Let $\bs q \in \real^n$ and $\bs p \in \real^n$ denote the generalized coordinate and the generalized momentum, respectively, the time evolution of system states in the phase space $(\bs q, \bs p)$ can be expressed as the symplectic gradient of the Hamiltonian function $H(\bs q, \bs p)$. 
For a large class of electromechanical systems, the Hamiltonian function admits a separable form with a quadratic kinetic energy, leading to 
%For this class of Hamiltonian systems, the dynamics can be expressed as 
\begin{equation}
\begin{aligned}
\dot{\bs q} = \nabla_{\bs p}H
\qquad
\dot{\bs p} = -\nabla_{\bs q}H,
\end{aligned}
\label{h-dynamics}
\end{equation}
with the following Hamiltonian function
\begin{align}
    H(\bs q, \bs p) &= \frac{1}{2}\bs p\tran \bs M^{-1}(\bs q) \bs p + V(\bs q), \label{hamiltonian}
\end{align}
where $\bs M(\bs q)$ is the mass/inertia matrix, and $V(\bs q)$ is the potential energy of the system. Also, since 
\begin{equation}
    \dot{H} = (\nabla_{\bs q}H)\tran \dot{\bs q} + (\nabla_{\bs p}H)\tran \dot{\bs p} = 0,
\end{equation}
the total energy of the system is always conserved along the system trajectories. 
%given by \eqref{h-dynamics} (we have dropped $H$'s explicit dependence on $\bs q$ and $\bs p$ for brevity). 
The Hamiltonian formulation extends to the port-Hamiltonian (pH) framework by incorporating inputs and dissipation:
\begin{align}
    \begin{bmatrix}
        \dot{\bs q} \\
        \dot{\bs p}
    \end{bmatrix} &=  
    \begin{bmatrix}
        \bs 0 & \bs I \\
        -\bs I & -\bs D(\bs q)
    \end{bmatrix}
    \begin{bmatrix}
        \nabla_{\bs q}H \\
        \nabla_{\bs p}H
    \end{bmatrix} + 
    \begin{bmatrix}
        \bs 0 \\
        \bs g(\bs q)
    \end{bmatrix} \bs u, \label{ph-dynamics} 
    \end{align}
   or equivalently, 
    \begin{align}
    \dot{\bs z} &= \bs F(\bs q) \nabla_{\bs z} H + \bs G(\bs q) \bs u, \label{ph_dynamics-concise}
\end{align}
where $\bs F = \bs J - \bs R \in \real^{2n \times 2n}$ is the system matrix, with $\bs J$ as the canonical symplectic matrix and $\bs R = \bs R\tran \succeq 0$ as the natural damping matrix taking into account the dissipative and friction effects, 
%(sometimes $\bs D = \bs D\tran \succeq 0$ which composes $\bs R$ is referred to as the damping/dissipation matrix), 
and $\bs g(\bs q) \in \real^{n\times m}$ is the input matrix assumed to be full column rank, $\text{rank}\{\bs g(\bs q)\}=m$.  The system is usually affine in the external input $\bs u \in \mc U \subset \real^m$, and it only affects the generalized momenta. It can be shown that $\bs F(\bs q) + \bs F\tran(\bs q) = -2\bs R \preceq 0$. Most robotic systems can be modeled as the system exchanging energy (or matter) with its environment or other systems, thus admitting a port-Hamiltonian formulation \citep{schaft1999l2}.

\subsection{Learning Hamiltonian Dynamics from Data} \label{sec-HNN}
Using Neural Networks (NNs) to learn a system represented by an ODE, $\dot{\bs z}=\bs f(\bs z, t)$, from observed data has gained significant attention in recent years. Physics-informed approaches, such as \citep{Zhong2020Symplectic}, learn Hamiltonian $H$ (or Lagrangian $L$) instead of directly parametrizing $\bs f$, and integrate the resulting dynamics using Neural ODE (NODE) solvers~\citep{chen2018neural}. Since $L$ or $H$ are scalar functions, they are easier to learn than vector fields, and their structure promotes energy conservation in the learned dynamics. 

Accordingly, the dynamics in \eqref{ph_dynamics-concise} are approximated by
\begin{equation}
    \bs f_{\theta} = \bs F_{\theta}(\bs q) \nabla H_{\theta} + \bs G_{\theta}(\bs q) \bs u.  \label{eq:f-theta-concise}
\end{equation}
The learned vector field $\bs f_{\theta}$ is integrated using a NODE solver to obtain predictions ${\bs Z}_\theta = \texttt{NODESolve}(\bs z_{t_0}, \bs f_{\theta}, \bs t)$, and $\dot{{\bs Z}}_\theta$, where ${\bs Z}_\theta = {\bs z_\theta}_{t_1}, {\bs z_\theta}_{t_2}, \ldots, {\bs z_\theta}_{t_n}$. Given observed data $\{\bs Z, \dot{\bs Z}\}$, the parameters $\theta$ are updated by minimizing the prediction error $\mathbb{L} = \|\dot{\bs Z} - \dot{{\bs Z_\theta}}\|^2_2$ using backpropagation through the NODE solver.

Incorporating control inputs can be handled by augmenting the system as
\begin{equation}
    \begin{bmatrix}
        \dot{\bs z} \\
        \dot{\bs u}
    \end{bmatrix}=
    \begin{bmatrix}
        \bs f_{\theta}(\bs z, \bs u) \\
        \bs h_{\theta}(\bs z)
    \end{bmatrix}= \Tilde{\bs f}_{\theta}(\bs z, \bs u),
\end{equation}
where $\bs h_{\theta}(\bs z) = 0$ if $\bs u$ in the observed data $(\bs z_t, \bs u_t)_{t=t_0}^{t_n}$ remains constant. 
This allows the use of separate neural networks to parameterize the system components, namely the mass matrix $\bs M_{\theta_1}(\bs q)$, potential $V_{\theta_2}(\bs q)$, input matrix $\bs g_{\theta_3}(\bs q)$, and dissipation $\bs D_{\theta_4}(\bs q)$. The resulting dynamics take the structured form
\begin{equation} \label{eq:f-theta}
    \bs f_{\theta}(\bs z, \bs u) =  
    \begin{bmatrix}
        \bs 0 & \bs I \\
        -\bs I & -\bs D_{\theta_4}(\bs q)
    \end{bmatrix}
    \begin{bmatrix}
        \nabla_{\bs q}H_{\theta} \\
        \nabla_{\bs p}H_{\theta}
    \end{bmatrix} + 
    \begin{bmatrix}
        \bs 0 \\
        \bs g_{\theta_3}(\bs q)
    \end{bmatrix} \bs u, %\label{ph-dynamics-theta}
\end{equation}
where
\begin{equation}
    H_{\theta}(\bs z) = \frac{1}{2}\bs p\tran \bs M_{\theta_1}^{-1}(\bs q) \bs p + V_{\theta_2}(\bs q).
\end{equation}

\begin{remark}
    Bounded activation functions (e.g., sigmoid, hyperbolic tangent) ensure bounded NN outputs, which in turn guarantee boundedness of $H_\theta$ and passivity of the unforced system.
\end{remark}

% \rev{
% \begin{proposition} \label{prop:1}
%     If parameterized $\bs f_\theta$ consists of NNs with $\gamma$-bounded activation functions $\sigma:\real \mapsto \real$, that is, $\exists$ $\gamma > 0$ such that for all $a \in \real$, $|\sigma(a)| \leq \gamma$, then $H_\theta$ is upper bounded for $\{\hat{\bs z}, \bs u\} \in \mc O \times \mc U$ (where $\mc O$ is a compact set) and \eqref{eq:f-theta} represents a pH system.
% \end{proposition}

% \begin{remark}
%     It is trivial to show that the output of a NN with $\gamma$-bounded activation functions (for example, sigmoid, hyperbolic tangent, etc.) is upper and lower bounded. In fact, it is enough to have $\gamma$-bounded activation functions in the last layer. A direct consequence of this is the lower boundedness of $H_\theta$ and thus passivity/stability of the unforced system \eqref{eq:f-theta}.
% \end{remark}
% }

Prior work \citep{Zhong2020Symplectic} demonstrates that such structured parameterizations can learn system dynamics to high accuracy.
%the dynamics of several systems can be learned up to numerical precision by using embedded angle data and parameterizing $\bs f_{\theta}$ in a systematic way by incorporating prior knowledge using Hamiltonians into end-to-end learning. 
It should be noted that NNs designed for the mass/inertia matrix $\bs M_{\theta_1}(\bs q)$ and the dissipation matrix $\bs D_{\theta_4}(\bs q)$ need to ensure their positive definiteness. Moreover, when $\bs p$ data is unavailable, the learned Hamiltonian may exhibit scaling ambiguities, often appearing as offsets in the mass matrix and learned potential.
% there exists a scaling invariance in the learned Hamiltonian, which might show itself as an offset in the learned $\bs M_{\theta_1}(\bs q)$.

\subsection{Optimal Energy-Shaping Control via EB-PBC}
We re-write the pH system \eqref{ph_dynamics-concise} with output equation
\begin{equation}   \label{eq:actual_system}
\begin{aligned}
    \dot{\bs z} = \bs F \nabla H + \bs G \bs u, \qquad
    \bs y = \bs G\tran \nabla H,
\end{aligned}
\end{equation}
satisfying the power balance equation
\begin{equation}
    \dot{H} = \bs u\tran \bs y - d,
\end{equation}
where $d = -\nabla H\tran \bs F \nabla H \geq 0$ is the natural %open-loop 
dissipation of the system.
The objective of designing a passivity-based control (PBC) with damping injection (DI) is to find a control $\bs u = \bs \beta(\bs z) + \bs v$ that renders the closed-loop system passive with respect to the desired stored energy (Hamiltonian) $H_d$ with equilibrium at desired $\bs z_d$ such that it satisfies the power-balance equation
\begin{equation}    \label{eq:desired_power_balance}
    \dot{H}_d = \bs y_d\tran \bs v - d^*,
\end{equation}
where $\bs y_d$ is the new passive output, and $d^*$ is the desired dissipation. For energy balancing (EB) PBC, we set the desired dissipation $d^* =d$,  the natural dissipation of the system.
%that is, $d^* = d = -\nabla H\tran \bs F \nabla H$.
A controller $\bs \beta(\bs z)$  achieves energy balancing if $\dot{H}_a = -\bs \beta\tran \bs y$, where 
$H_a := H_d - H$ is the added energy. 
%the state-feedback control $\bs \beta(\bs z)$ is EB if . 
Such a controller is defined by
\begin{equation}
    \bs \beta(\bs z) = -\bs G^{\dagger} \bs F\tran \nabla_{\bs z} H_a \label{u-EB}
\end{equation}
%solves the EB-PBC problem 
with $\bs G^{\dagger}$ as the left pseudo-inverse of $\bs G$ and $H_a$ satisfying the following matching PDEs \citep{ortega2008control}:
\begin{equation}
    \begin{bmatrix}
        \bs G^{\perp}  \bs F\tran\\
        \bs G\tran
    \end{bmatrix} \nabla_{\bs z} H_a = 0. \label{eq:pde-constraints}
\end{equation}
Here $\bs G^{\perp}$ is the left full-rank annihilator of $\bs G$, that is, $\bs G^{\perp}\bs G = 0$. 

Damping injection is achieved via %passive output $\bs y_d$ as
%To inject damping, feeding back the new passive output $\bs y_d$ as
\begin{equation}
    \bs v = - \bs K \bs y_d = - \bs K \bs G\tran \nabla_{\bs z}H_d,    \label{u-DI}
\end{equation}
where $\bs K = \bs K\tran \succ 0$, the damping matrix, ensures passivity by adding dissipation to the closed-loop system.

The choice of a desired energy $H_d$ and damping $\bs K$ is not unique, and an optimal EB-PBC policy %for the  $\bs u$ 
chooses optimal $H_d$ and $\bs K$ with respect to some cost function, while satisfying the nonlinear PDE constraint \eqref{eq:pde-constraints}. 
To address this issue, $H_d$ and $\bs K$ are parameterized using neural networks, enabling optimization over a structured policy class that implicitly satisfies the PDE constraints~\citep{massaroli2022optimal}.
Let $V_{\phi_V}^*$ and $\bs K^*_{\phi_K}$ be parametrized NN representations of the desired constant damping matrix and the  added potential, such that 
% we parametrize 
% the desired constant damping matrix $\bs K^*$, and the  added potential $V^*(\bs q)$  by NNs $V_{\theta_V}^*$ and $\bs K^*_{\theta_K}$, respectively, such that 
\[
{H_\phi}_{d}(\bs q, \bs p) = \bs p\tran \bs M^{-1} \bs p/2 + V(\bs q) + V_{\phi_V}^*(\bs q),
\]
and the resulting parameterized control policy, with $\phi := (\phi_V, \phi_K)$ is
\begin{align}
    \bs u_{\phi}(t, \bs z) = -\bs G^{\dagger}(\bs q) \bs F\tran(\bs q) \nabla_{\bs z} V_{\phi}^*(\bs q) - \bs K^*_{\phi}(t, \bs z) \bs g\tran(\bs q) \nabla_{\bs p}H(\bs z). \label{eq:u-theta-f}
\end{align}
The optimal policy $\bs u_{\phi}^*$ is obtained by minimizing a performance cost  $l_{\bs u}(\bs z_0, \bs u_{\phi})$ using gradient-based optimization, with gradients computed via backpropagation through the system dynamics.

% \bigskip 

% To address the issue of satisfying the nonlinear PDE constraint \eqref{eq:pde-constraints} and searching for $H_d$ and $\bs K^*$ in a functional space,   $H_d$ and $\bs K^*$ are parametrized using NNs and optimal control policy is computed on a submanifold that implicitly satisfies the PDE constraints \citep{massaroli2022optimal}.
% To learn the desired energy $H_d$ and 

% desired constant damping matrix $\bs K^*$, added potential $V^*(\bs q)$ (which fully compensates $V(\bs q)$ and adds a desired $V_d(\bs q)$) and damping matrix $\bs K^*$ are represented by NNs $V_{\theta_V}^*$ and $\bs K^*_{\theta_K}$, respectively, such that the ${H_d}_{\theta}(\bs q, \bs p) = \bs p\tran \bs M^{-1} \bs p/2 + V(\bs q) + V_{\theta}^*(\bs q)$, with $\phi := (\theta_V, \theta_K)$. 

% Denoting the performance measure as the cost function $l_{\bs u}(\bs z_0, \bs u_{\theta})$, we pose our parameterized control policy as
% \begin{align}
%     \bs u_{\phi}(t, \bs z) = & -\bs G^{\dagger}(\bs q) \bs F\tran(\bs q) \nabla_{\bs z} V_{\phi}^*(\bs q) \notag \\
%     & - \bs K^*_{\phi}(t, \bs z) \bs g\tran(\bs q) \nabla_{\bs p}H(\bs z). \label{eq:u-theta-f}
% \end{align}
% The optimal policy $\bs u_{\phi}^*$ can then be found out by iteratively updating the parameters $\theta$ using SGD on the gradient of performance measure d$l_u/$d$\phi$ calculated by backpropagating through the NODE used to propagate the system dynamics forward.

\begin{remark}
   Control strategies that completely compensate the natural potential $V(\bs q)$ and replace it with the desired $V_d(\bs q)$ suffer from 
   sensitivity to modeling errors and may discard beneficial passive dynamics of the system. Allowing partial shaping enables the controller to exploit intrinsic system structure and reduce control effort.  
    % The class of control strategies that seek to completely compensate the natural potential $V(\bs q)$ and replace it with the desired $V_d(\bs q)$ suffers from notable shortcomings, such as, imperfect knowledge of $V(\bs q)$ introducing unwanted nonlinearities stemming from partial or imperfect compensation. It also prevents the controller from leveraging the useful passive elements of the system that could assist in achieving the desired system behavior and potentially help reduce the control effort/energy.
\end{remark}

\subsection{Problem Formulation} \label{sec:prob_formulation}
% Model-based policy gradient methods require a model to sample the system trajectories. In the absence of a system model, the dynamics can be learned from the observed data for port-Hamiltonian systems, as shown in Sec \ref{sec-HNN}. We can thus break down this problem into two sub-problems,
% \begin{itemize}
%     \item \text{Sub-Problem 1:} Learning the system model using trajectory data, and
%     \item \text{Sub-Problem 2:} Using the learned system, learn an optimal EB-PBC policy of the form \eqref{eq:u-theta-f} that stabilizes the true system. 
% \end{itemize}
% These two sub-problems can then be solved iteratively or using an end-to-end learning scheme. We detail our proposed approach in the next section.

Accurate system models are essential for synthesizing high-performance controllers. However, in many practical settings, the system dynamics are unknown or only partially known and must be learned from data. For port-Hamiltonian systems, this can be achieved using structure-preserving learning methods, as discussed in Section~\ref{sec-HNN}.
We consider the problem of co-learning a port-Hamiltonian system model and synthesizing an optimal energy-shaping controller from trajectory data. The goal is to obtain a control policy of the form \eqref{eq:u-theta-f} that stabilizes the true system while optimizing performance, despite model uncertainty.

A key challenge is that model accuracy and control performance are tightly coupled: the quality of the learned model affects the policy, while the policy determines the regions of the state space from which data is collected. To address this, we propose a coupled learning framework that alternates between model refinement using policy-induced trajectories and policy optimization on the updated model. We detail our proposed approach in the next section.

\section{Proposed Physics Informed Architecture} \label{sec:proposedPIDPG}

In this section, we present the proposed learning architecture, which jointly learns the port-Hamiltonian system model and the optimal EB-PBC policy through an alternating optimization procedure (Fig.~\ref{fig:flowchart}).

In the absence of a true system model, the learned dynamics \eqref{eq:f-theta-concise} are used within the EB-PBC policy, yielding
\begin{align}     \label{eq:u-theta-f-theta}
    \bs u_{\phi}(t, \bs z) = -\bs G^{\dagger}_\theta(\bs q) \bs F_\theta\tran(\bs q) \nabla_{\bs z} V_{\phi}^*(\bs q) - \bs K^*_{\phi}(t, \bs z) \bs g_\theta\tran(\bs q) \nabla_{\bs p}H_\theta(\bs z).
\end{align}

A naive approach would be to perform end-to-end optimization by jointly updating the system and policy parameters $(\theta,\phi)$ using a composite loss $(\mb L + l_u)$. However, this approach is problematic for several reasons. First, the control loss $l_u$ induces gradients on $\theta$ that prioritize closed-loop performance over dynamical fidelity, potentially violating the physical consistency of the learned model and compromising robustness guarantees. Second, randomly initialized models generally do not satisfy the PDE matching constraints \eqref{eq:pde-constraints}, rendering the EB-PBC formulation invalid during early training. Third, jointly optimizing identification and control objectives introduces well-known gradient conflicts in multi-objective learning~\citep{bischof2025multi, wang2021understanding}, for which no general solution exists.
To address these challenges, we adopt an alternating optimization scheme with policy-aware data collection:

%Naively, one could attempt end-to-end optimization by back-propagating a composite loss through both the system and control parameters $(\theta, \phi)$ simultaneously. End-to-end optimization of a composite loss $(\mb L + l_{\bs u})$ on $(\theta, \phi)$ simultaneously is problematic because the control loss $l_{\bs u}$ imposes gradients on $\theta$ that favor low policy cost rather than dynamical fidelity, compromising the physical consistency of the learned model and, consequently, any robustness guarantees. Additionally, (i) randomly initialized dynamics do not satisfy the PDE matching constraints~\eqref{eq:pde-constraints}, invalidating the EB-PBC formulation~\eqref{eq:u-theta-f} until the model has sufficiently converged; and (ii) balancing identification and control losses in a single objective introduces well-documented gradient pathologies~\citep{bischof2025multi, wang2021understanding}, which is an active research area with no one-size-fits-all solution.  We circumvent these issues through alternating optimization between the two sub-problems, with policy-aware data collection coupling them across iterations: \\

\noindent 
\textbf{Phase 0 (Warm-up).} System trajectories are generated by applying constant excitation inputs $\bs u \in
\{-2, -1, 0, 1, 2\}$ to the true system. The system model $\bs f_\theta$ is trained on this data to obtain an initial pH
model that satisfies the structural properties required for EB-PBC synthesis.

\noindent 
\textbf{Phase 1 (Alternating Optimization).} At each outer iteration $k=1,2,\cdots :$

\begin{enumerate} [label=(\alph*)]
    \item \textit{System model update ($\theta$-step)}: The system model $\bs f_\theta$ is refined using trajectory data from the true system. This data is a mixture of (i) step-excited trajectories that ensure broad coverage of the state space, and (ii) policy-excited trajectories generated under the current EB-PBC $\bs u_\phi$, which focuses model capacity on the region of state space actually traversed by the closed-loop system. During this step, only the system parameters $\theta$ are updated, and the policy parameters $\phi$ are frozen.
    \item \textit{Policy update ($\phi$-step)}: Using the updated model $\bs f_\theta$, the EB-PBC parameters $\phi$ are optimized by integrating a batch of initial conditions $z_0 \sim \mb P_{z_0}$ through a NODE solver for $t \in [0, T]$ and minimizing the policy cost    
\begin{align} \label{policy_cost}
    l_u(\bs z_0) = L(T, \phi, \bs z_T(\bs z_0, \bs u_{\phi})) + \eta \int_0^T \left| \bs u_{\phi} (t, \bs z) \right| \textrm{d}t, \qquad \eta > 0.
\end{align}
During this step, only $\phi$ is updated, and $\theta$ is frozen.
\end{enumerate}

The use of policy-excited data in step (a)  is the key mechanism coupling model learning and control synthesis: as the controller improves, the data distribution shifts toward the policy-relevant region, which in turn improves the model precisely where closed-loop accuracy matters. This structure ensures that at every iteration,
(i) the system model $\bs f_\theta$ satisfies the pH structure (required by Lemma~\ref{prop:1}), (ii) the EB-PBC $\bs u_\phi$ operates on a valid pH manifold and is thus inherently passive (Lemma~\ref{lemma:2}), and (iii) the model training error $\epsilon = \norm{\bs f_\theta - \bs f}$ decreases preferentially in the policy-relevant region.
The loss gradients are computed via PyTorch automatic differentiation (AD) \citep{baydin2018automatic} library and backpropagated using the adjoint method \citep{chen2018neural}.

\subsubsection*{Sim-to-real Robustness}

A central question is whether the learned controller \eqref{eq:u-theta-f-theta} stabilizes the true system \eqref{ph_dynamics-concise}, i.e., whether it is robust to model mismatch, e.g., due to the sim-to-real gap. To address this, we augment the policy loss with a dissipation regularization term
\begin{align}   \label{eq:diss_regu}
    \subscr{\mc L}{diss} = \subscr{\lambda}{diss} \expt_{\bs z}\left[\text{ReLU}\left( \dot{H_{\theta, \phi}}_d(\bs z, \bs u_\phi) + \rho \norm{\bs z - {\bs z}_d}^2 \right) \right].
\end{align}
This term enforces a strict energy dissipation in the closed-loop system away from the critical point ${\bs z}_d$ of ${H_{\theta, \phi}}_d$. Here $\dot{H_{\theta, \phi}}_d = \nabla {H_{\theta, \phi}}_d\tran \bs f_\theta(\bs z, \bs u_\phi)$, $\rho$ is a tunable dissipation rate, and $\subscr{\lambda}{diss}$ is the regularization weight. Our augmented policy loss thus becomes
\begin{equation}
    \mb L_u = l_u(\bs z_0) + \subscr{\mc L}{diss}.
\end{equation}
In the next section, we show how this augmented policy cost helps impart sim-to-real robustness to the learned EB-PBC $\bs u_\phi$.

\section{Stability Analysis}    \label{sec:analysis}

In this section, we show that the optimal EB-PBC controller trained on the learned system $\bs f_\theta$ can stabilize the true system $\bs f$, thereby providing robustness to model mismatch. 
We first establish structural properties of the learned port-Hamiltonian model that will be used in the subsequent stability analysis.

\begin{lemma}\label{prop:1}
If $\bs f_\theta$ is parameterized using neural networks with bounded activation functions, then $H_\theta$ is bounded on compact domains, and \eqref{eq:f-theta} represents a port-Hamiltonian system.
\end{lemma}
\begin{proof}
    The proof is straightforward and is omitted for brevity. 
\end{proof}

We now establish the structural and stability properties of the EB-PBC law under the actual model.
\begin{lemma}\label{lemma:1}
Consider the port-Hamiltonian system \eqref{eq:actual_system} under the parameterized EB-PBC law \eqref{eq:u-theta-f}. Suppose the added potential $V_\phi^* \in \mc C^1$, the damping matrix $K_\phi^*(t,\bs z) \succeq \kappa I, \kappa > 0$ uniformly in $(t,\bs z)$ and is bounded, and the matching conditions \eqref{eq:pde-constraints} hold. Then the closed-loop system can be written as
\begin{equation}
    \dot{\bs z} = \bs F_d(t,\bs z)\nabla H_{\phi_d}(\bs z),
    \label{eq:closed-loop-Fd}
\end{equation}
with $\bs F_d(t,\bs z) := \bs F(\bs z) - \bs G(\bs z) K_\phi^*(t,\bs z) \bs G(\bs z)^\top$, $H_{\phi_d} = H + V_\phi^*$, and
\begin{equation}
    \dot H_{\phi_d} = - \nabla H_{\phi_d}^\top
    \left(
        \bs R + \bs G K_\phi^* \bs G^\top
    \right)
    \nabla H_{\phi_d}
    \le 0.
    \label{eq:Hd-dot-lemma}
\end{equation}
Moreover, if $H_{\phi_d}$ is bounded below along system trajectories and $\dot H_{\phi_d}$ is uniformly continuous, then $\dot H_{\phi_d}(t)\to 0$ as $t\to\infty$. If, in addition, $\bs R+\bs G K_\phi^* \bs G^\top \succ 0$, uniformly in $t$,
in a neighborhood of $z_d$, and $z_d$ is a strict local minimum of $H_{\phi_d}$, then $z_d$ is locally asymptotically stable.
\end{lemma}

\begin{proof}
Let $H_a := H_{\phi_d} - H$. From matching conditions \eqref{eq:pde-constraints} 
\[
\bs G^\top \nabla H_{\phi_a} = 0,
\qquad
\bs G^\perp \bs F^\top \nabla H_{\phi_a} = 0,
\]
hence $\bs F^\top \nabla H_{\phi_a} \in \mathrm{im}(\bs G)$. With $\beta = -\bs G^\dagger \bs F^\top \nabla H_{\phi_a}$, it follows that
\[
\bs F^\top \nabla H_{\phi_a} = -\bs G \beta.
\]
Premultiplying by $\nabla H_{\phi_a}^\top$ and using $\bs F=\bs J-\bs R$, $\bs J^\top=-\bs J$, yields
\[
\nabla H_{\phi_a}^\top \bs R \nabla H_{\phi_a} = 0 \;\Rightarrow\; \bs R \nabla H_{\phi_a} = 0.
\]

Next, using $\nabla H = \nabla H_{\phi_d} - \nabla H_{\phi_a}$ and $\bs u = \beta + \bs v$,
\[
\dot{\bs z}
= \bs F(\nabla H_{\phi_d} - \nabla H_{\phi_a}) + \bs G\beta + \bs G\bs v.
\]
Using $\bs F+\bs F^\top = -2\bs R$ and $\bs R\nabla H_{\phi_a}=0$, we obtain
\[
-\bs F\nabla H_{\phi_a} = \bs F^\top \nabla H_{\phi_a},
\]
hence
\[
\dot{\bs z}
= \bs F \nabla H_{\phi_d} + \bs F^\top \nabla H_{\phi_a} + \bs G\beta + \bs G\bs v.
\]
Using $\bs F^\top \nabla H_{\phi_a} = -\bs G\beta$, the terms cancel and
\[
\dot{\bs z}
= \bs F \nabla H_{\phi_d} + \bs G\bs v.
\]
With $\bs v = -K_\phi^* \bs G^\top \nabla H_{\phi_d}$,
\[
\dot{\bs z}
=
\left(\bs F - \bs G K_\phi^* \bs G^\top\right)\nabla H_{\phi_d},
\]
which proves \eqref{eq:closed-loop-Fd}.

Finally, to establish convergence, we note, along trajectories,
\[
\dot H_{\phi_d}
=
\nabla H_{\phi_d}^\top \dot{\bs z}
=
\nabla H_{\phi_d}^\top
\left(\bs F - \bs G K_\phi^* \bs G^\top\right)
\nabla H_{\phi_d}.
\]
Using $\bs F=\bs J-\bs R$ and skew-symmetry of $\bs J$,
\[
\dot H_{\phi_d}
=
-\nabla H_{\phi_d}^\top
\left(
\bs R + \bs G K_\phi^* \bs G^\top
\right)
\nabla H_{\phi_d}
\le 0.
\]
Hence $H_{\phi_d}(t)$ is non-increasing and converges to a finite limit. If $\dot H_{\phi_d}$ is uniformly continuous, then by Barbalat's lemma,
\[
\lim_{t\to\infty}
\nabla H_{\phi_d}^\top
\left(
\bs R + \bs G K_\phi^* \bs G^\top
\right)
\nabla H_{\phi_d}
=0.
\]

Finally, if $\bs R+\bs G K_\phi^* \bs G^\top \succ 0$ locally, then $\nabla H_{\phi_d}(\bs z(t))\to 0$. Since $z_d$ is a strict local minimum of $H_d$, it is the only critical point in a neighborhood, and thus $\bs z(t)\to z_d$, proving local asymptotic stability.
\end{proof}

The following result shows that the EB-PBC structure is preserved under the learned model.

\begin{lemma}   \label{lemma:2}
Let the parameterized damping $\bs K^*_{\phi}(t, \bs z) \succeq \kappa I, \kappa >0$ uniformly in $(t, \bs z)$, and the added potential $V^*_{\phi}$, be bounded. Additionally, assume that $\bs R_\theta + \bs G_\theta K_\phi^* \bs G_\theta^\top \succ 0$, uniformly in $t$, in a neighborhood of $\bs z_d$, and that $\bs z_d$ is a strict local minimum of $H_{{\theta, \phi}_d}$. Then the closed-loop learned system
    \begin{equation}
        \dot{\bs z} = \bs F_{{\theta, \phi}_d}\nabla H_{{\theta, \phi}_d}
    \end{equation}
    is asymptotically stable, with $\bs F_{{\theta, \phi}_d} = \bs F_\theta - \bs G_\theta \bs K^*_{\phi} \bs G_\theta\tran$ and $H_{{\theta, \phi}_d}$ satisfying the desired power-balance equation \eqref{eq:desired_power_balance}.
\end{lemma}

\begin{proof}
From Lemma~\ref{prop:1}, the learned dynamics $\bs f_\theta$ retain the port-Hamiltonian structure. The result then follows directly by applying Lemma~\ref{lemma:1} to the learned system.
\end{proof}
Thus, the optimal EB-PBC controller renders the closed-loop learned system \eqref{eq:f-theta-concise} passive with respect to the desired Hamiltonian and consequently renders the equilibrium locally asymptotically stable under the stated conditions.

Before moving on to our main result, we introduce the following useful result.

\begin{lemma} \label{prop:F_invertible}
Let $\bs F = \bs J - \bs R$, where $\bs J^\top = -\bs J$ and $\bs R = \bs R^\top \succeq 0$. If $\bs J$ is invertible, then $\bs F$ is invertible.
\end{lemma}

\begin{proof}
Let $\bs F\bs v = 0$. Then $\bs J\bs v = \bs R\bs v$. Premultiplying by $\bs v^\top$ gives
\[
\bs v^\top \bs J \bs v = \bs v^\top \bs R \bs v.
\]
Since $\bs J$ is skew-symmetric, the left side is zero, so $\bs v^\top \bs R \bs v = 0$. As $\bs R \succeq 0$, this implies $\bs R\bs v = 0$, hence $\bs J\bs v = 0$. Since $\bs J$ is invertible, $\bs v = 0$.
\end{proof}

We state the assumptions used in the subsequent analysis.

\begin{enumerate}[label=\text{(A\arabic*)}]
\item \emph{Model approximation error:}
There exists a domain $\mc D \subset \real^{2n}$ such that, for all $\bs z \in \mc D$ and admissible inputs $\bs u$, including $\bs u = \bs u_\phi(\bs z)$.
\begin{align*}
\|\bs f_\theta(\bs z,\bs u) - \bs f(\bs z,\bs u)\|
\le \epsilon,\qquad
\|\bs G_\theta(\bs z) - \bs G(\bs z)\|
\le \epsilon_G.
\end{align*}

    \item \emph{Damping approximation:}
    \[
    \|\bs R_\theta(\bs z) - \bs R(\bs z)\| \le \epsilon_R,
    \quad \forall \bs z \in \mc D.
    \]

    \item \emph{Gradient boundedness:}
There exists $\bar{\alpha}_{\mc D}>0$ such that
\[
\|\nabla H_\theta(\bs z)\| \le \bar{\alpha}_{\mc D}, \quad
\|\nabla V_\phi^*(\bs z)\| \le \bar{\alpha}_{\mc D}, \quad \forall \bs z\in\mc D.
\]

    \item \emph{Shared energy shaping:}
    The same learned potential $V_\phi^*(\bs q)$ is applied to both systems, i.e.,
    \[
    H_a(\bs z) = H_{\phi_a}(\bs z) = V_\phi^*(\bs q).
    \]

\item \emph{Approximate dissipation of the trained learned closed-loop:}
There exist constants $\rho>0$ and $\varepsilon_{\mathrm{diss}}\ge 0$ such that, for all $\bs z\in\mc D$,
\begin{align*}
 \nabla H_{{\theta, \phi}_d}^\top
\bigl(\bs R_\theta(\bs z) + \bs G_\theta(\bs z)\bs K_\phi^*(t, \bs z)\bs G_\theta(\bs z)^\top\bigr)
\nabla H_{{\theta, \phi}_d}
\ge
\rho \|\bs z-\bs z_d\|^2 - \varepsilon_{\mathrm{diss}}.   
\end{align*}
\end{enumerate}

Assumptions A1–A2 quantify the approximation accuracy of the learned port-Hamiltonian model and its dissipation structure, which can be achieved using sufficient data coverage. Assumption A3 is mild and holds on compact domains due to the smoothness of the learned Hamiltonian. Assumption A4 reflects the controller implementation, where the same energy-shaping policy is applied to both the learned and true systems. 
Assumption A5 characterizes the dissipative behavior of the trained learned closed-loop system. It imposes a lower bound on the learned dissipation up to a residual $\varepsilon_{\mathrm{diss}}$, reflecting approximation and optimization errors. This condition is promoted by the dissipation regularization term \eqref{eq:diss_regu}, which penalizes violations of energy decay during training. 
%The parameter $\rho$ acts as a tunable decay rate, allowing the learned controller to enforce stronger dissipation and thereby reduce the ultimate bound on the true system trajectories.

% \begin{lemma}   \label{lemma:3}
%     Let $\bs F=\bs J - \bs R$ be the system matrix with $c_F=\subscr{\sigma}{min}(\bs F) > 0$. Suppose
%     \begin{enumerate}[label=(\roman*)]
%         \item $\norm{\bs f_\theta(\bs z, \bs u) - \bs f(\bs z, \bs u)} < \epsilon$ for sufficiently small $\epsilon > 0$,
%         \item $\norm{\bs R_\theta - \bs R} < \epsilon_R$ holds for sufficiently small $\epsilon_R >0$,
%         \item $\bs z_\theta$ lies in a compact set $\mc C \in \real^{2n}$ containing the critical point ${\bs z_\theta}_d$ of $H_\theta$ such that $\norm{\nabla H_\theta} \leq \bar{\alpha}_{\mc C}$ for some constant $\bar{\alpha}_{\mc C} > 0$, and
%         \item the same EB-PBC NN $V^*_\phi$ is applied to both the learned and the true systems at identical states, so that $H_a(\bs z) = {H_\theta}_a(\bs z) = V^*_\phi(\bs q)$.
%     \end{enumerate}
%     Then the closed-loop system Hamiltonians are such that
%     \begin{equation}
%         \norm{\nabla H_d - \nabla {H_\theta}_d} \leq \frac{1}{c_F} (\epsilon + \epsilon_R \bar{\alpha}_{\mc C}) := \delta(\epsilon, \epsilon_R).
%     \end{equation}
%     In particular, $\delta \to 0$ as $\epsilon , \epsilon_R \to 0$.
% \end{lemma}

\begin{lemma}\label{lemma:3}
Let $\bs F = \bs J - \bs R$ and define $c_F := \inf_{\bs z \in \mc D} \sigma_{\min}(\bs F(\bs z)) > 0$.
Under Assumptions A1–A4, for all $\bs z \in \mc D$,
\[
\|\nabla H_d(\bs z) - \nabla H_{{\theta, \phi}_d}(\bs z)\|
\le
\frac{1}{c_F}\bigl(\epsilon + \epsilon_R \bar{\alpha}_{\mc D}\bigr)
=: \delta(\epsilon,\epsilon_R).
\]
In particular, $\delta \to 0$ as $\epsilon,\epsilon_R \to 0$.
\end{lemma}
\begin{proof}
For $u=0$,
\begin{align}
\norm{\bs F_\theta \nabla H_\theta - \bs F \nabla H}
= \norm{\bs f_\theta(\bs z, 0) - \bs f(\bs z, 0)} < \epsilon.
\end{align}
Writing $\bs F_\theta = \bs F - \tilde{\bs R}$ with $\tilde{\bs R} := \bs R_\theta - \bs R$,
\[
\bs F \nabla H - \bs F_\theta \nabla H_\theta = \bs F(\nabla H - \nabla H_\theta) + \tilde{\bs R}\nabla H_\theta.
\]
Thus,
\[
\norm{\bs F \bs h} \le \epsilon + \epsilon_R \bar{\alpha}_{\mc D},
\quad \bs h := \nabla H - \nabla H_\theta.
\]
Using $\norm{\bs F^{-1}} \le 1/c_F$,
\[
\norm{\bs h} \le \frac{1}{c_F}(\epsilon + \epsilon_R \bar{\alpha}_{\mc D}).
\]
Since $H_{{\theta, \phi}_d} - H_\theta = H_d - H = V^*_\phi(\bs q)$, the result follows.
\end{proof}

We are now ready to state our main result.

\begin{theorem}\label{theorem1}
Consider the true system \eqref{ph_dynamics-concise} in closed loop with the EB-PBC controller $\bs u_\phi$ trained on the learned model $\bs f_\theta$. Let $\bs z_d$ be a strict local minimum of $H_d$.

Under Assumptions A1–A4, there exists a neighborhood $\mc O \subset \mc D$ of $\bs z_d$ and a nonnegative residual $\xi=\mathcal{O}(\epsilon,\epsilon_R,\epsilon_G)$ such that, for all $\bs z \in \mc O$,
\[
\dot H_d(\bs z)
\le
-\nabla H_d(\bs z)^\top
\bigl(\bs R(\bs z)+\bs G(\bs z)\bs K_\phi^*(\bs z)\bs G(\bs z)^\top\bigr)
\nabla H_d(\bs z)
+\xi.
\]

Consequently, the true closed-loop system is dissipative with respect to $H_d$ up to the residual $\xi$.

Furthermore, if Assumption A5 holds, then
\[
\dot H_d(\bs z)
\le
-\rho\|\bs z-\bs z_d\|^2+\xi+\varepsilon_{\mathrm{diss}},
\quad \forall \bs z \in \mc O.
\]

Hence, the closed-loop system is locally practically asymptotically stable: trajectories are driven to the smallest sublevel set of $H_d$ containing
\[
\left\{\bs z\in\mc O:\ \|\bs z-\bs z_d\|\le \sqrt{\frac{\xi+\varepsilon_{\mathrm{diss}}}{\rho}}\right\},
\]
and remain in this sublevel set thereafter.
\end{theorem}

\begin{proof}
Consider $\dot{\bs z}=\bs f(\bs z,\bs u_\phi)$ with Lyapunov function $H_d$. Then
\begin{align}
\dot{H}_d
&= \langle \nabla H_d, \bs f \rangle \notag\\
&= \langle \nabla H_{{\theta, \phi}_d}, \bs f_\theta \rangle
+ \langle \nabla H_d - \nabla H_{{\theta, \phi}_d}, \bs f_\theta \rangle + \langle \nabla H_{{\theta, \phi}_d}, \bs f - \bs f_\theta \rangle
+ \langle \nabla H_d - \nabla H_{{\theta, \phi}_d}, \bs f - \bs f_\theta \rangle.
\end{align}

From Lemma~\ref{lemma:2}, first term satisfies
\[
\langle \nabla H_{{\theta, \phi}_d}, \bs f_\theta \rangle
=
-\nabla H_{{\theta, \phi}_d}^\top
\bigl(\bs R_\theta + \bs G_\theta \bs K_\phi^* \bs G_\theta^\top\bigr)
\nabla H_{{\theta, \phi}_d}.
\]

Using $\|\bs R_\theta-\bs R\|\le \epsilon_R$ and $\|\bs G_\theta-\bs G\|\le \epsilon_G$, and boundedness of $\bs G$, $\bs G_\theta$, and $\bs K_\phi^*$ on $\mc O$, we have
\[
\|\bs G_\theta \bs K_\phi^* \bs G_\theta^\top - \bs G \bs K_\phi^* \bs G^\top\|
\le C_1 \epsilon_G.
\]
The second term 
\[
\langle \nabla H_d - \nabla H_{{\theta, \phi}_d}, \bs f_\theta \rangle
\le
C_2 \delta(\epsilon,\epsilon_R),
\]
for some constants $C_1, C_2 < \infty$ depending on bounds of $\bs f_\theta$, $\nabla H_{{\theta, \phi}_d}$, and $\bs K_\phi^*$ on $\mc D$.
The third and fourth terms satisfy
\begin{align*}
\langle \nabla H_{{\theta, \phi}_d}, \bs f - \bs f_\theta \rangle \le
2\epsilon \bar{\alpha}_{\mc D}, \quad \text{and} \quad
\langle \nabla H_d - \nabla H_{{\theta, \phi}_d}, \bs f - \bs f_\theta \rangle \le
\delta(\epsilon,\epsilon_R)\epsilon.
\end{align*}
Combining all terms,
\[
\dot H_d
\le
-\nabla H_{{\theta, \phi}_d}^\top
(\bs R_\theta + \bs G_\theta \bs K_\phi^* \bs G_\theta^\top)
\nabla H_{{\theta, \phi}_d}
+ \xi,
\]
where $\xi=\mathcal{O}(\epsilon,\epsilon_R,\epsilon_G)$.

If Assumption A5 holds, we get
% \[
% \nabla H_{\theta d}^\top
% (\bs R_\theta + \bs G_\theta \bs K_\phi^* \bs G_\theta^\top)
% \nabla H_{\theta d}
% \ge
% \rho \|\bs z-\bs z_d\|^2 - \varepsilon_{\mathrm{diss}}.
% \]
% Using Lemma~\ref{lemma:3} and continuity of the quadratic form under model perturbations,
% \[
% \nabla H_d^\top
% (\bs R + \bs G \bs K_\phi^* \bs G^\top)
% \nabla H_d
% \ge
% \rho \|\bs z-\bs z_d\|^2 - \varepsilon_{\mathrm{diss}} - \xi.
% \]
% Thus,
\[
\dot H_d
\le
-\rho \|\bs z-\bs z_d\|^2 + \xi + \varepsilon_{\mathrm{diss}}.
\]
Hence, trajectories are driven to the smallest sublevel set of $H_d$ containing the stated neighborhood and remain in this sublevel set thereafter, establishing local practical asymptotic stability.
\end{proof}

\begin{remark}[Role of $\rho$ and practical stability]
The parameter $\rho$ acts as a tunable dissipation gain in the learned closed-loop system and plays a central role in determining the size of the ultimate bound. In particular, increasing $\rho$ tightens the ultimate bound and improves convergence toward the desired equilibrium.
\end{remark}

\section{Simulation Study} \label{sec:Experiments}
In this section, we discuss the details of the simulation study and the results.

\subsection{{Simulated System and Data Generation}}
We validate the proposed learning framework by learning an optimal EB-PBC controller for stabilization and swing-up control of a planar pendulum and a torsional pendulum system from their trajectory data. To show the efficacy and generalization of our approach, we choose the mass, length, and input matrix of the planar pendulum from a uniform distribution $\mb U(0,2)$. This is done to ensure the robustness of our approach to the pendulum parameters. Dissipation $\bs D$ is set to $0$ for the planar pendulum experiments. The observed trajectory data (true system data) is generated using the true Hamiltonian function \eqref{hamiltonian}. $512$ initial conditions for training trajectories (for both system and control policy) are sampled from a random uniform distribution of angles and angular velocities $\bs q_0, \dot{\bs q}_0$ as $\mathbb{P}_{\bs z_0} \sim \mathbb{U}([-2\pi, 2\pi] \times [-\pi, \pi])$.

\subsection{{Training Details}}
For learning the system model, we minimize the integrated error over training trajectories. For each initial condition $(\bs z_0, t_0)$ in the training dataset, the $\mc L_2$ loss is calculated as $\mathbb{L} = \|\dot{\bs Z} - \dot{\hat{\bs Z}}\|^2_2$, where $\hat{\bs Z} = \texttt{NODESolve}(\bs z_{t_0}, \bs f_\theta, t)$ comes from the learned model $\bs f_\theta$, and we obtain $\dot{\hat{\bs Z}}$ directly using $\bs f_\theta$, and $\{\bs Z, \dot{\bs Z}\} = (\bs z_t, \dot{\bs z}_t, u_t)_{t=t_0}^{t_n}$ is the (observed) trajectory data generated from the true system model (equivalent to getting data from the true system). The training trajectories were only $T=0.15s$ long, whereas the learned model is evaluated for $3s$ to show the efficacy and generalization well beyond the training trajectory lengths.
The cost function used to learn the stabilization control policy is
\begin{align} \label{policy-cost-full}
    \mb E_{\bs z_0}\left[l_{\bs u}(\bs z_0)\right] =  -\frac{1}{N} \sum_{i=1}^N \textrm{log} \Upsilon^*\left(\bs z_T(\bs z_0^i, \bs u_{\phi})\right) +\frac{\eta}{N} \sum_{i=1}^N \int_0^T \left| \bs u_{\phi}(\bs z_t(\bs z_0^i, \bs u_{\phi})) \right| \textrm{d}t,
\end{align}
where the target system configuration is defined as a probability set $\Upsilon^*(\bs z)$ centered at $(q^*, p^*)$ with variance $\sigma^2 \bs I$, and $\eta > 0$. The system trajectories are generated by integrating the initial conditions (sampled from $\mb P_{\bs z_0}$) through a NODE solver, and the cost gradient d$\mb L_{u}/$d$\phi$ was computed by backpropagating using the adjoint method \citep{chen2018neural}. For the swing-up control, we augment the control policy cost function as
\begin{align}   \label{eq:swing_up_cost}
    \mb E_{\bs z_0} \left[l_{\bs u}(\bs z_0)\right] =  \frac{1}{N} \sum_{i=1}^N \left( (\bs z_0^i - \bs z^*)\tran Q (\bs z_0^i - \bs z^*) \right) +\frac{\eta}{N} \sum_{i=1}^N \int_0^T \left[ \bs u_{\phi}(\bs z_t(\bs z_0^i, \bs u_{\phi})) \right]^2 \textrm{d}t,
\end{align}
where $Q$ was set to $\texttt{diag}([100,10])$. For the dissipation regularization term, we used $\subscr{\lambda}{diss} = 1$, and $\rho = 1e-3$.

Unless specified otherwise, the ODE solver used for NODE is the fourth-order Runge-Kutta (\texttt{rk4}), discretization step used is $10^{-2}$ for system training and $2\times 10^{-2}$ for controller training, \textsc{Adam} optimizer used for gradient descent step with the learning rate set to $10^{-3}$ and weight decay to $10^{-4}$. Learning rate cosine annealing \citep{loshchilov2017sgdr} is used for training the system model $\bs f_\theta$. Other hyperparameter values used are: $\eta = 10^{-3}, \sigma^2 = 10^{-3}$.

\begin{figure}
    \centering
    \begin{subfigure}{0.23\linewidth}
	    \centering
        \includegraphics[width=1\linewidth, height=1.1\linewidth, keepaspectratio]{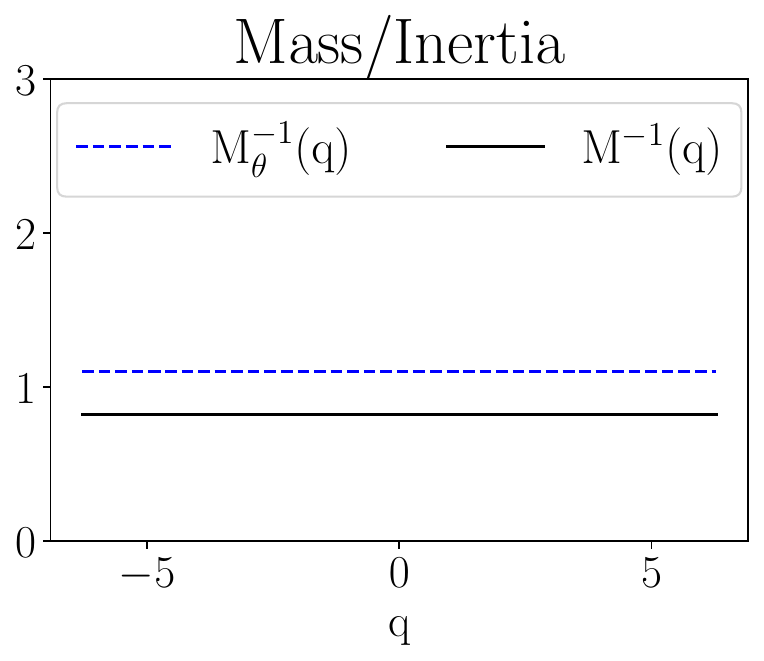}
        \caption{}
        \label{mass}
    \end{subfigure} \hspace{-1.3em}
    ~
    \begin{subfigure}{0.24\linewidth}
	    \centering
        \includegraphics[width=1\linewidth, height=1.1\linewidth, keepaspectratio]{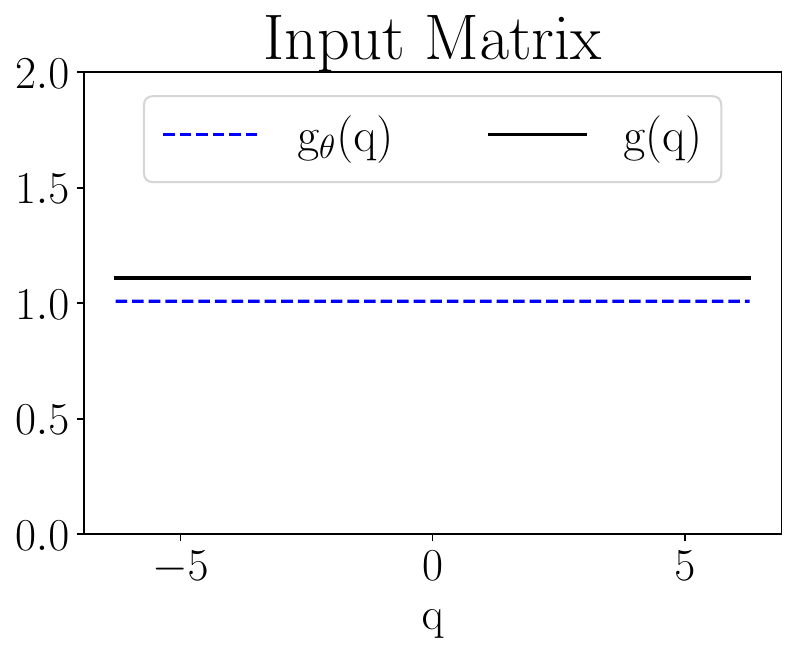}
        \caption{}
        \label{gq}
    \end{subfigure} \hspace{-1.3em}
    ~
    \begin{subfigure}{0.25\linewidth}
	    \centering
        \includegraphics[width=1\linewidth, height=1.1\linewidth, keepaspectratio]{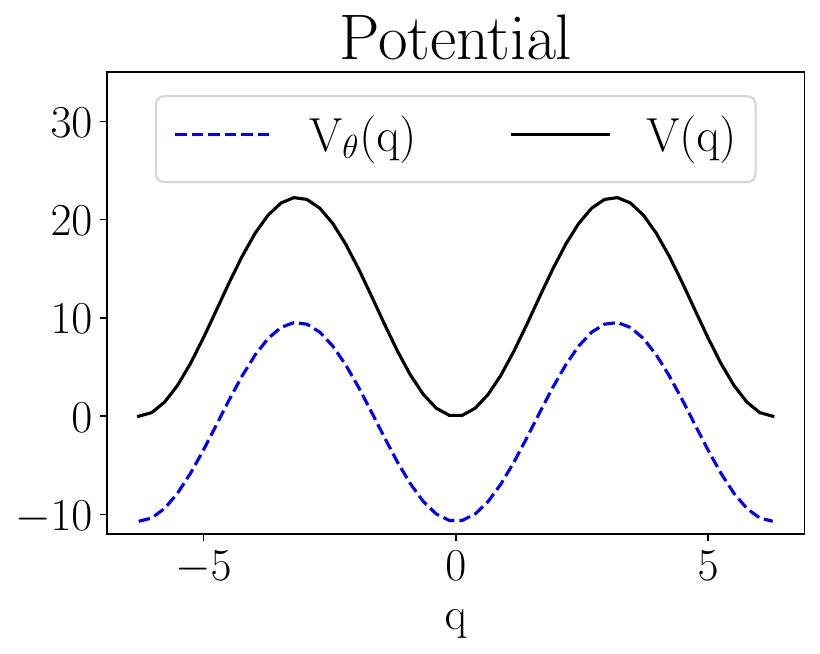}
        \caption{}
        \label{v_chnn}
    \end{subfigure} \hspace{-1.3em}
    ~
    \begin{subfigure}{0.25\linewidth}
	    \centering
        \includegraphics[width=1\linewidth, height=1.1\linewidth, keepaspectratio]{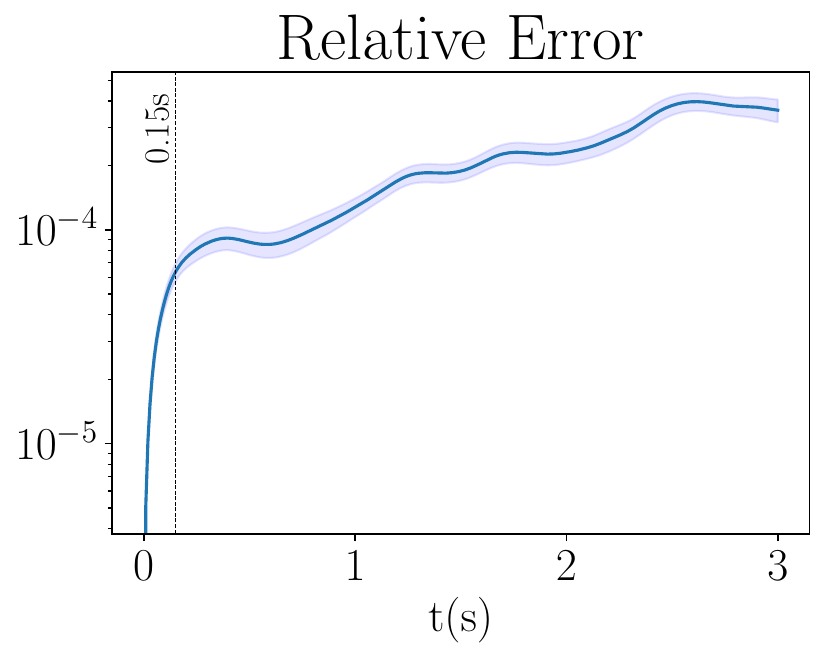}
        \caption{}
        \label{fig:rel_err}
    \end{subfigure} \hspace{-1.3em}
    ~
    \caption{\textbf{Learned system matrices for $1$-link planar pendulum:} (a),(b), and (c) show that $M_{\theta_1}, V_{\theta_2}(q), g_{\theta_3}(q)$ match the ground truth with an offset (invariance owing to training on $\dot{q}$ instead of $p$). (d) The relative error in states shows rollout for $3s$ with $95\%$ confidence bands; the y-axis is log-scaled for clearer representation.}
    \label{fig:f-theta}
\end{figure}

\subsection{{NN Architectures}}
Here, we detail the different NN architectures used. All NNs used are fully connected, and we describe them in the following format. 
\[
\subscr{d}{in}-n_1\supscr{f}{act}_1 - \cdots -n_N\supscr{f}{act}_N-\subscr{d}{out},
\]
where $\subscr{d}{in}$ and $\subscr{d}{out}$ are input and output dimensions, $n_i$ is the dimension of layer $i$ and $\supscr{f}{act}_i$ is the associated activation function.
\begin{itemize}
    \item $M_{\theta_1}^{-1}= L_{\theta_1}\tran L_{\theta_1}: 2 - 300\ \mathrm{tanh} - 300\ \mathrm{tanh} - 300\ \mathrm{tanh} - 1$
    \item $g_{\theta_2}: 2 - 400 \; \mathrm{tanh} - 400 \;\mathrm{tanh} - 1$
    \item $V_{\theta_3}: 2 - 50\; \mathrm{tanh} - 50 \;\mathrm{tanh} - 1$
    \item $V^*_{\phi}: 1 - 64 \;\mathrm{softplus} - 64 \;\mathrm{softplus} - 64\; \mathrm{tanh} - 1$
    \item $K^*_{\phi}: 3 - 64 \; \mathrm{softplus} - 64\; \mathrm{softplus} - 1 \;\mathrm{softplus}$
\end{itemize}
The $\mathrm{softplus}$ activation is appended to the output of $K^*_{\phi}$ to enforce the positive definiteness of the learned desired damping matrix.

\subsection{Stabilization and Swing-up Control of Planar Pendulum System}

To show the convergence of a planar pendulum system parameterized by NNs, we compare learned system matrices: $M_{\theta_1}^{-1}(q), V_{\theta_2} (q), g_{\theta_3}(q)$ with the ground truth in Fig.~\ref{fig:f-theta}. The learned $M_{\theta}$, $g_{\theta}(q)$, and $V_\theta(q)$ match the ground truth with a constant offset, which is not an issue for control as we still learn the $q$ dynamics correctly. The offset in the learned $V_{\theta}(q)$ is not an issue either since the potential is a relative notion and only the derivative of $V_{\theta}$ plays a role in the dynamics of the system.
Error in prediction trajectories (simulated for a much longer time than the training trajectories) is computed as a relative error defined by $\mathrm{Err}(t) = \norm{\bs z(t) - \hat{\bs z}(t)}_2 / \norm{\bs z(t) + \hat{\bs z}(t)}_2$. Fig.~\ref{fig:rel_err} shows the relative error in states averaged over $512$ initial conditions sampled from $\mb P_{\bs z_0}$. The shaded regions are $95\%$ confidence intervals, and the dotted line (at $0.15$s) shows the length of the training trajectories. It can be seen that the trained model can accurately predict trajectories for a much longer horizon than the duration of training trajectories. 
Phase plots in Fig. \ref{fig:eb-pbc} show that the resulting closed-loop system is stabilized by the proposed controller for the target configuration $(q^*, p^*) = (0,0)$ (here the pendulum angle $q$ is measured relative to the negative $y$-axis).
\begin{figure}[h!]
    \centering
    \includegraphics[width=0.5\linewidth]{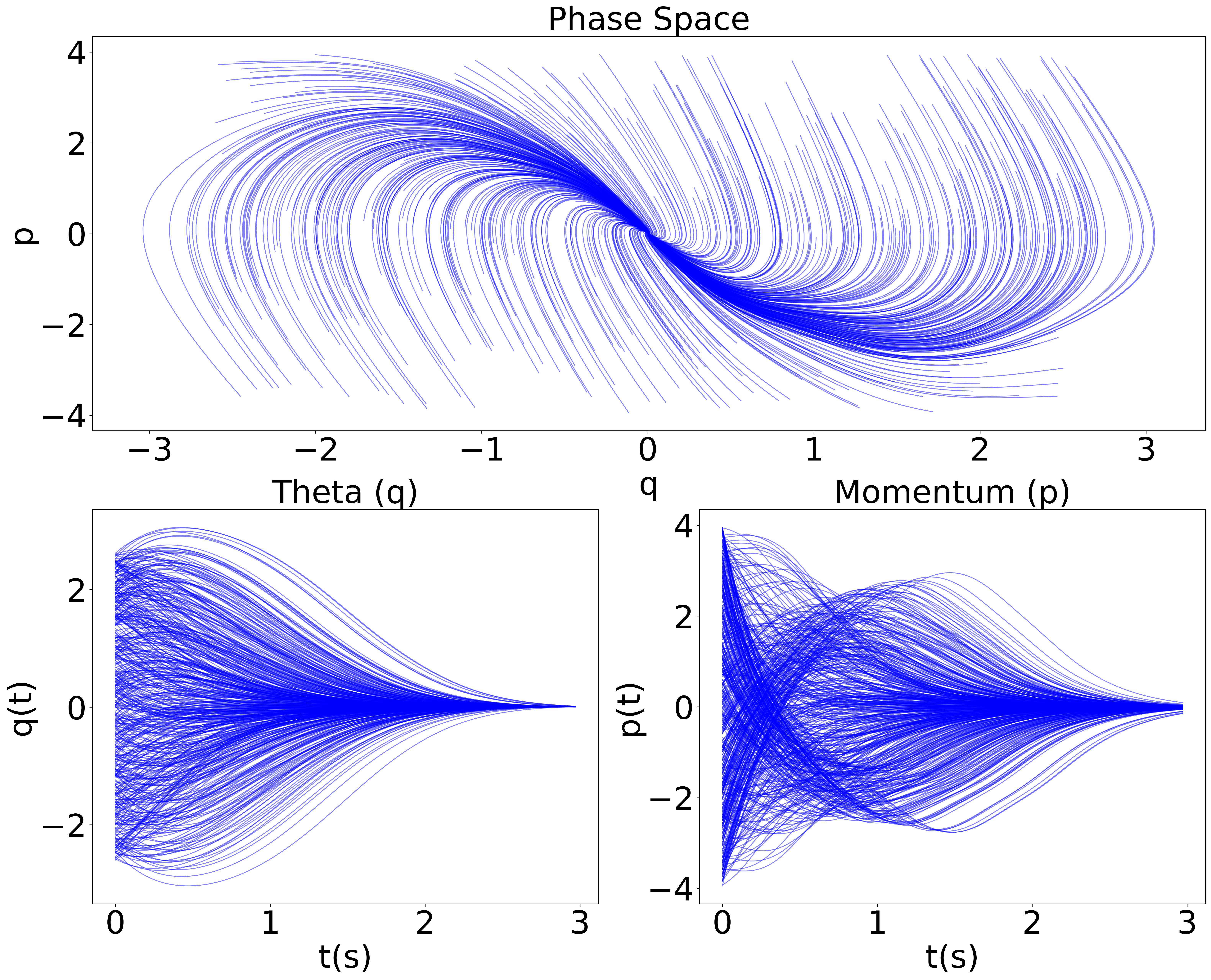}
    \caption{\textbf{State-Regulation by Optimal EB-PBC on $1$-link Planar Pendulum:} Phase plots of the learned system trajectories regulated to $\bs z=[0,0]$ for $512$ initial conditions sampled from the distribution $\mb P_{\bs z_0}$.
    }
    \label{fig:eb-pbc}
\end{figure}
\begin{figure}[h!]
    \centering
    \begin{subfigure}{0.49\linewidth}
	    \centering
        \includegraphics[width=1\linewidth, height=1\linewidth, keepaspectratio]{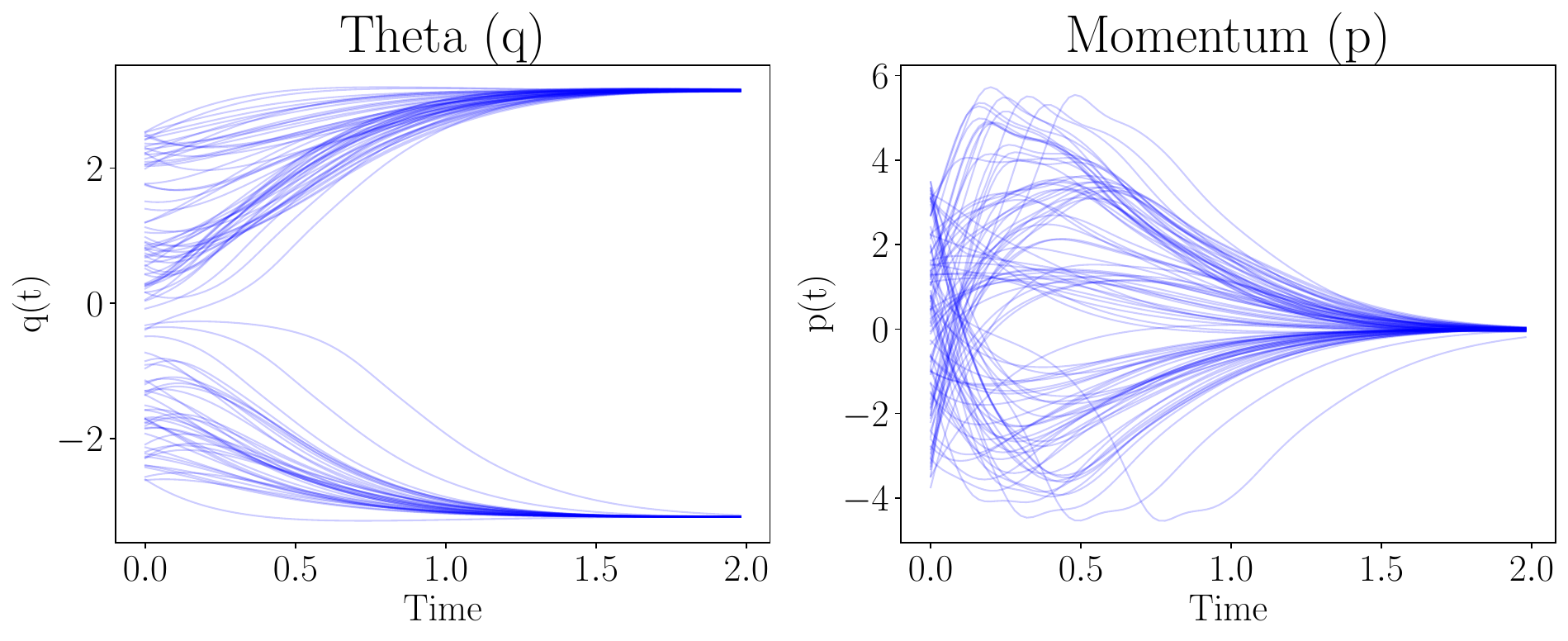}
        \caption{}
        \label{SU_PIDPG}
    \end{subfigure} \hspace{-1.3em}
    \begin{subfigure}{0.49\linewidth}
	    \centering
        \includegraphics[width=1\linewidth, height=1\linewidth, keepaspectratio]{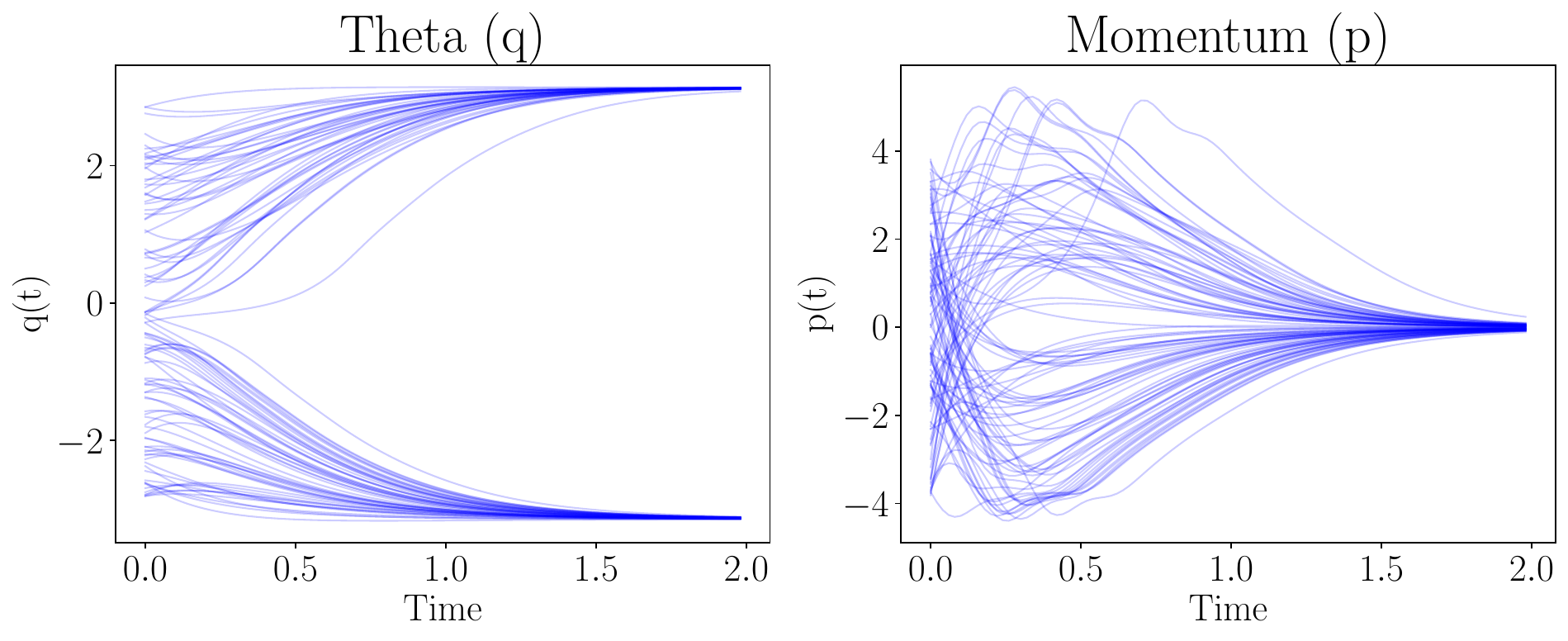}
        \caption{}
        \label{SU_PH}
    \end{subfigure}
    ~
    \caption{\textbf{Swing-up Optimal EB-PBC for $1$-link Planar Pendulum:} Optimal EB-PBC controlled system trajectories for \textbf{(a)} learned system, and \textbf{(b)} true system for $100$ initial conditions sampled from $\mb P_{\bs z_0}$}
    \label{fig:swing_up}
\end{figure}
\begin{figure}[h!]
    \centering
    \includegraphics[width=0.5\linewidth, height=0.8\linewidth, keepaspectratio]{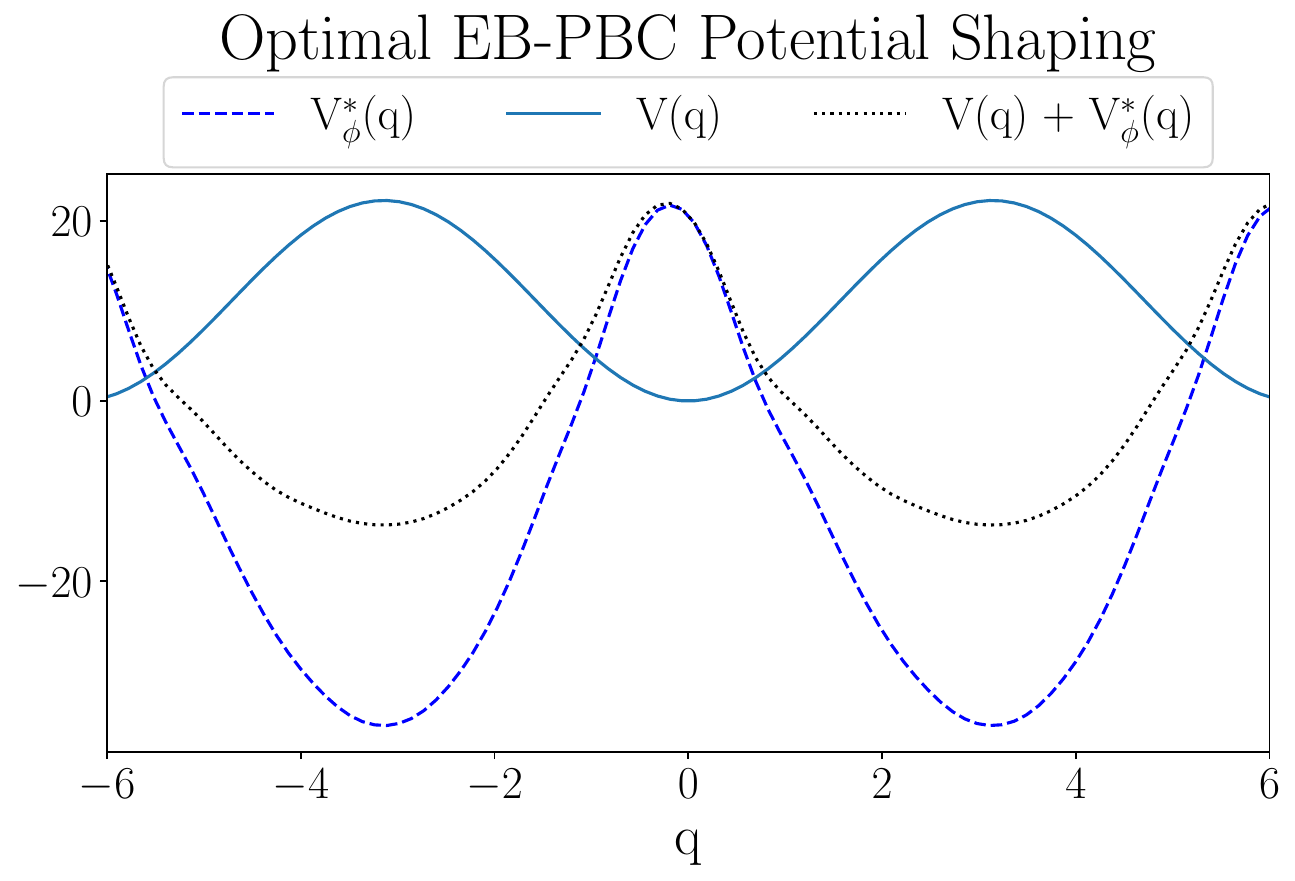}
    \caption{\textbf{Energy Shaping for Swing-up Control of $1$-link Planar Pendulum:} The proposed framework learns the added desired potential $V_{\phi}^*$ to place the minimum at $\pm\pi$ for pendulum swing-up.}
    \label{fig:ES_potential}
\end{figure}
Fig. \ref{SU_PIDPG} shows the time evolution of states $(q, p)$ for 100 initial conditions sampled from $\mb P_{\bs z_0}$ to stabilize the learned pendulum system at the inverted position, whereas Fig. \ref{SU_PH} shows the optimal EB-PBC performance on the true pendulum model. These phase plots show that the proposed framework drives the pendulum to the inverted position from a wide range of initial conditions.
Fig. \ref{fig:ES_potential} shows that the learned added potential $V_{\phi}^*$ reshapes the total energy landscape so that the resulting potential has its minimum at $\pm\pi$, the upright position. Although the training data has trajectory samples in the range $(-\pi,\pi)$, $V_{\phi}^*$ extrapolates well.
\begin{figure}[h!]
    \centering
    \begin{subfigure}{0.45\linewidth}
	    \centering
        \includegraphics[width=1\linewidth, height=1.1\linewidth, keepaspectratio]{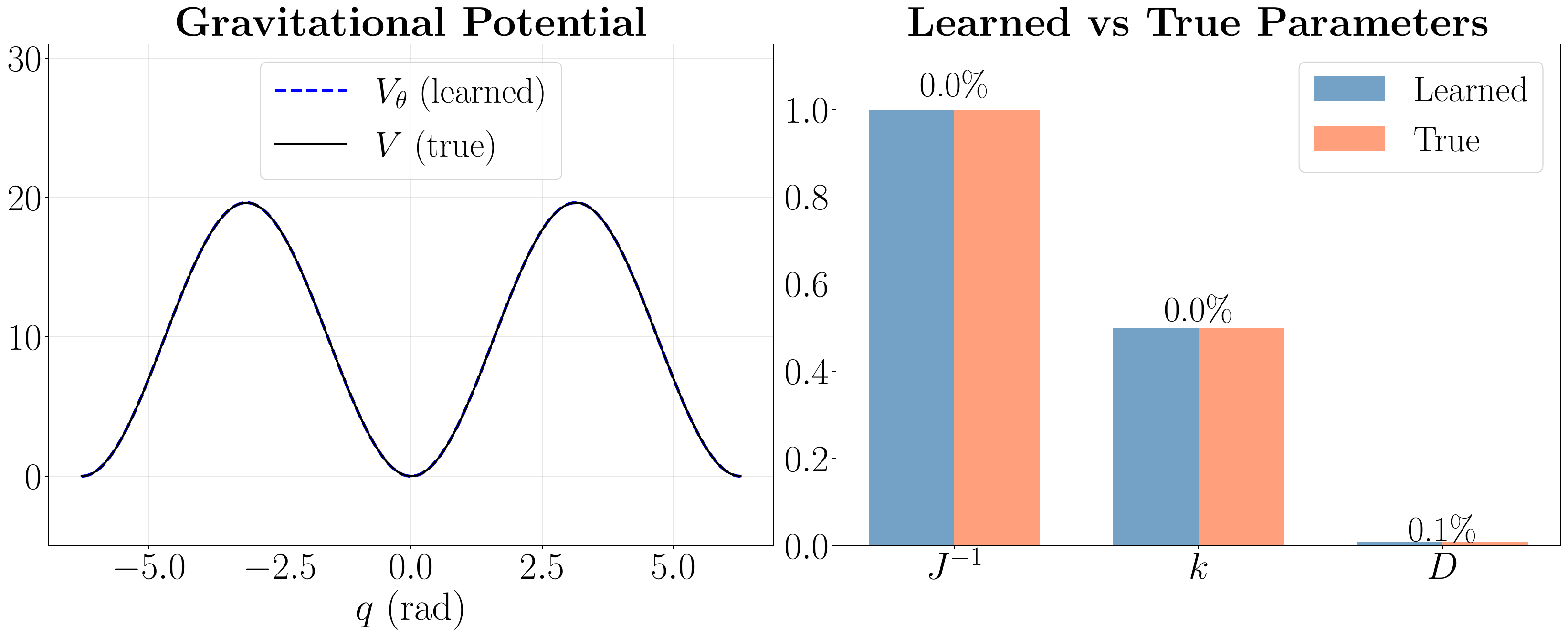}
        \caption{}
        \label{fig:tpend_symoden}
    \end{subfigure}
    ~
    \begin{subfigure}{0.49\linewidth}
	    \centering
        \includegraphics[width=1\linewidth, height=1.1\linewidth, keepaspectratio]{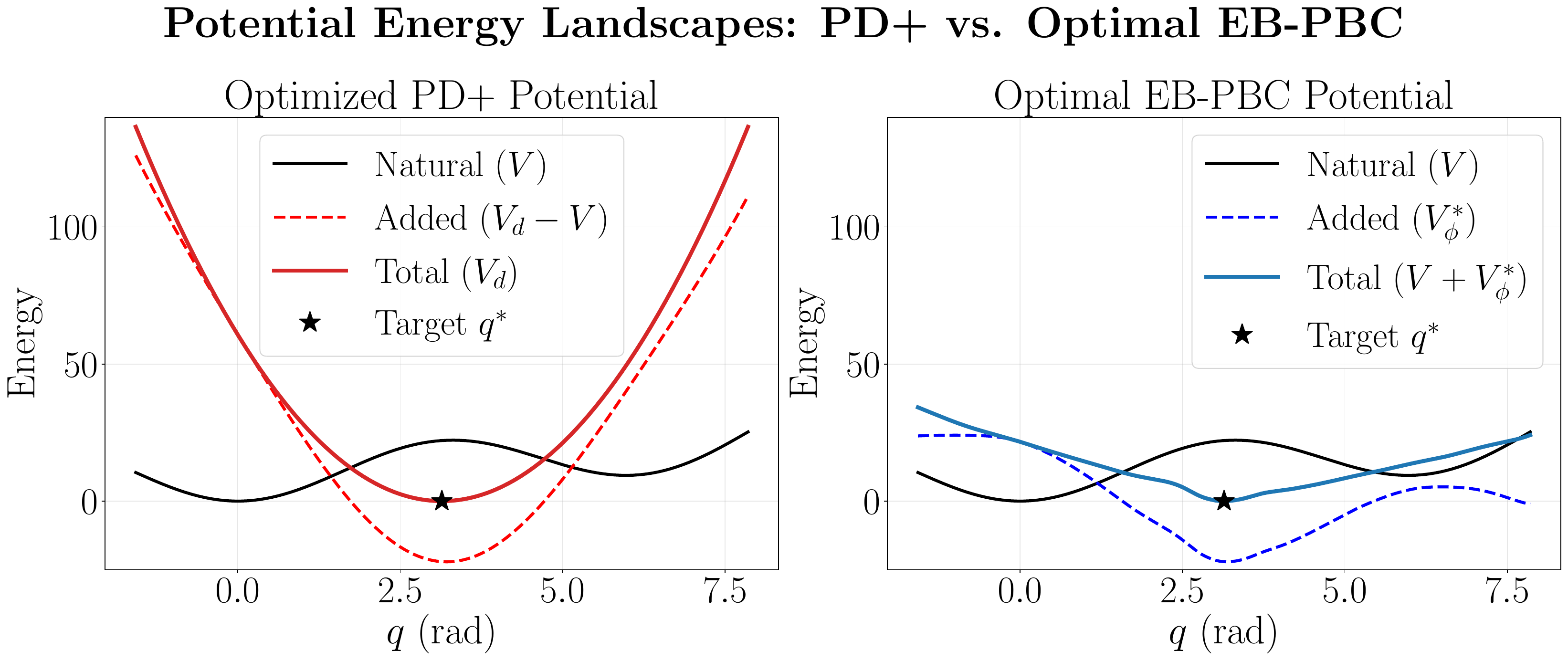}
        \caption{}
        \label{fig:oes_pd_potential}
    \end{subfigure} \\
    ~
    \begin{subfigure}{1\linewidth}
	    \centering
        \includegraphics[width=1\linewidth, height=1.1\linewidth, keepaspectratio]{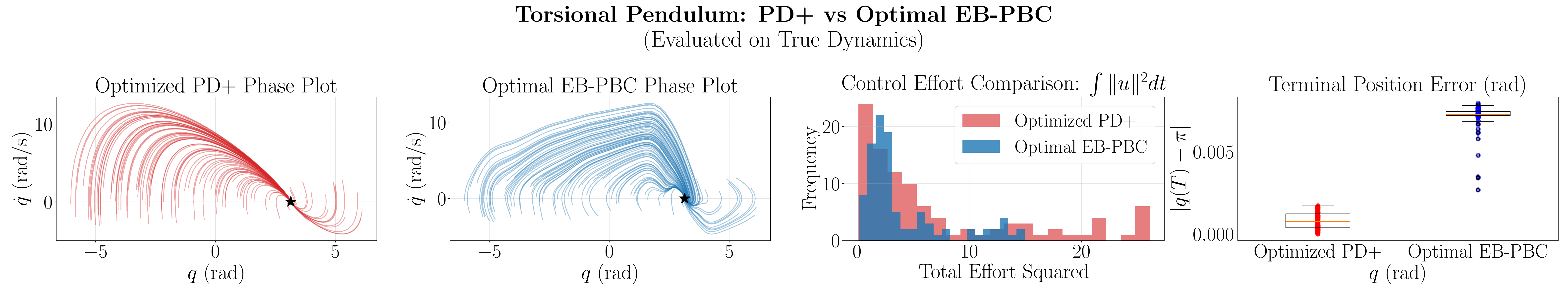}
        \caption{}
        \label{fig:oes_pd_traj}
    \end{subfigure}
    \caption{\textbf{Optimal EB-PBC on $1$-link Torsional Pendulum:} (a) shows the learned pH system with perfectly converged system parameters. Reshaped potential landscape in (b) shows how optimal EB-PBC learns to utilize the natural potential to achieve the minimum at $q^* = \pi$ compared to the PD+ controller that cancels and adds a quadratic potential. (c) The performance of optimal EB-PBC versus the optimized PD+ control for swing-up of the torsional pendulum shows how optimally reshaped potential leads to massive savings in control effort at the expense of slightly larger terminal position error.}
    \label{fig:oes_pd_comparison}
    \vspace{-0.05in}
\end{figure}
\begin{figure}[h!]
    \centering
    \includegraphics[width=1\linewidth]{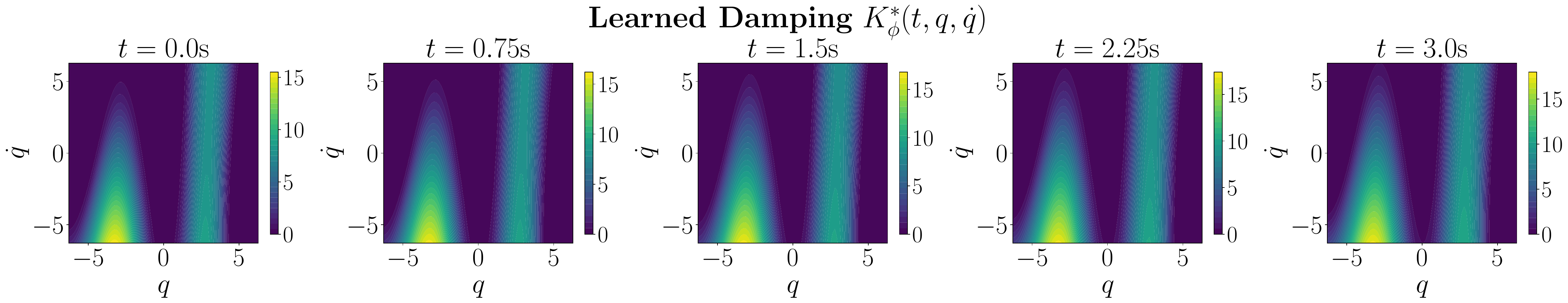}
    \caption{Snapshots of the learned state-dependent damping $K_\phi^*(t,q, \dot{q})$ over time for swing-up control of torsional pendulum: Rather than applying a constant uniform damping, the optimal EB-PBC dynamically adapts the dissipation profile, allowing the system to build momentum where necessary, and selectively increasing damping near the target equilibrium $(q^*=\pi, \dot{q}^*=0)$ to ensure smooth asymptotic stabilization.}
    \label{fig:oes_damping}
\end{figure}
\subsection{Comparing Optimal EB-PBC with PD + Potential Compensation for Swing-up Control of Torsional Pendulum}
We now show the performance of the optimal EB-PBC controller for the swing-up of a torsional pendulum. The Hamiltonian of this system is
\begin{equation}
    H(q, \dot{q}) = \frac{1}{2}J\dot{q}^2 + mgr(1 - \cos q) + \frac{1}{2}kq^2,
\end{equation}
with $J=1$, $m=1, r=1, k = 0.5$, and the damping coefficient $D=0.01$. Fig.~\ref{fig:tpend_symoden} shows the learned system parameters and the natural gravitational potential landscape. The constant offset in learned model parameters is resolved by assuming knowledge of the system potential at $q=0$ and anchoring the learned $V_\theta$. We then compare our proposed optimal EB-PBC control to the classic PD with potential compensation controller (PD+). We optimized the hyperparameters of the PD+ controller for a fair comparison. The PD+ control law is given by
\begin{equation}
    u_\phi(t, \bs z) = -mgr \sin q-kq + k_p(q - q^*) - k_d \dot{q},
\end{equation}
where we optimize $(k_p, k_d)$ with the swing-up cost \eqref{eq:swing_up_cost}. All the other hyperparameters were the same as for the optimal EB-PBC case. The optimal parameters for PD+ controller were obtained as $k_p=12.29, k_d = 6.30$. Fig.~\ref{fig:oes_pd_potential} shows how optimal EB-PBC learns to take advantage of the spring potential and gravity to considerably lower the control effort for swing-up as compared to PD+ case, which simply cancels the natural system potential and adds a quadratic potential around $q^*=\pi$. Fig.~\ref{fig:oes_pd_traj} shows the phase plot, total control effort (binned across various initial conditions across the complete trajectory), and terminal position error for both controllers. The mean effort expended by PD+ control was $7.10$ compared to $4.03$ for optimal EB-PBC for swing-up control of the torsional pendulum system. However, the proposed controller exhibits a slightly larger terminal position error, possibly owing to the flatter total potential landscape. The learned damping injection coefficient also has a non-trivial distribution as a function of system states and time. Contrary to the PD+ control, which assumes constant damping, the optimal EB-PBC controller damping depends on time (giving a feedforward control component), the sign of the velocity (asymmetric around $\dot{q}=0$), and the position, which is guided by the learned potential landscape, as shown in Fig.~\ref{fig:oes_damping}. This substantially generalizes the damping-injection implementation in PD+ and vanilla EB-PBC controllers.

\subsection{Comparing Optimal EB-PBC with PD + Potential Compensation for Swing-up Control of $2$-Link Torsional Pendulum}
To evaluate the proposed framework on a higher-dimensional system with rich passive elements, we compare the optimal EB-PBC controller against an optimized standard EB-PBC for a $2$-link rigid pendulum with parallel joint elasticity. The standard EB-PBC baseline is constructed by strictly compensating the natural gravitational and elastic potentials and imposing a quadratic desired potential, with its derivative and proportional gains optimized to minimize the same target cost function~\eqref{eq:swing_up_cost}.

% \subsubsection{Port-Hamiltonian System Learning}
The Hamiltonian of this system, expressed in terms of the generalized coordinates $\bs q=[q_1, q_2]\tran$ and velocities $\dot{\bs q}=[\dot{q}_1, \dot{q}_2]\tran$, is given by
\begin{equation}
    H(\bs q, \dot{\bs q}) = \frac{1}{2} \dot{\bs q}\tran \bs M(\bs q) \dot{\bs q} + V_g(\bs q) + V_e(\bs q),
\end{equation}
where the kinetic energy is characterized by the symmetric positive-definite inertia matrix $\bs M(\bs q) \in \real^{2\times 2}$ with elements
\begin{align}
    \bs M_{11}(\bs q) &= m_1 l_{c_1}^2 + m_2(l_1^2 + l_{c_2}^2) + I_1 + I_2 + 2m_2l_1l_{c_2}\cos(q_2), \\
    \bs M_{12}(\bs q) = \bs M_{21}(\bs q) &= m_2l_{c_2}^2 + I_2 + m_2l_1l_{c_2}\cos(q_2), \\
    \bs M_{22}(\bs q) &= m_2l_{c_2}^2 + I_2,
\end{align}
where $m_i, l_i, l_{c_i},$ and $I_i$ denote the mass, length, center of mass position, and moment of inertia of link $i\in\{1,2\}$, respectively.

The gravitational potential energy $V_g(\bs q)$ is defined as
\begin{equation}
    V_g(\bs q) = -(m_1l_{c_1}+m_2l_1)g\cos(q_1) - m_2l_{c_2}g\cos(q_1 + q_2).
\end{equation}
The elastic potential energy $V_e(\bs q)$ introduced by parallel torsional springs is
\begin{equation}
    V_e(\bs q) = \frac{1}{2}K_{s1}(q_1 - \subscr{q}{eq1})^2 + \frac{1}{2}K_{s2}(q_2 - \subscr{q}{eq2})^2,
\end{equation}
where $K_{s_i}$ represents the stiffness of the torsional spring at joint $i$ and $\subscr{q}{eq i}$ denote the corresponding equilibrium resting angle.
\begin{figure}[h!]
    \centering
    \begin{subfigure}{1\linewidth}
        \centering
        \includegraphics[width=1\linewidth, keepaspectratio]{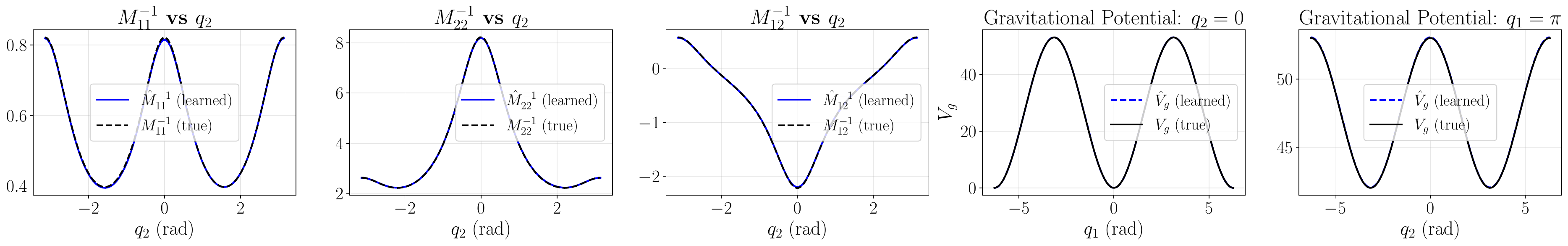}
        \caption{}
        \label{fig:M_V_2link}        
    \end{subfigure} \hspace{-1.3em}
    \\
    \begin{subfigure}{0.5\linewidth}
        \centering
        \includegraphics[width=1\linewidth, keepaspectratio]{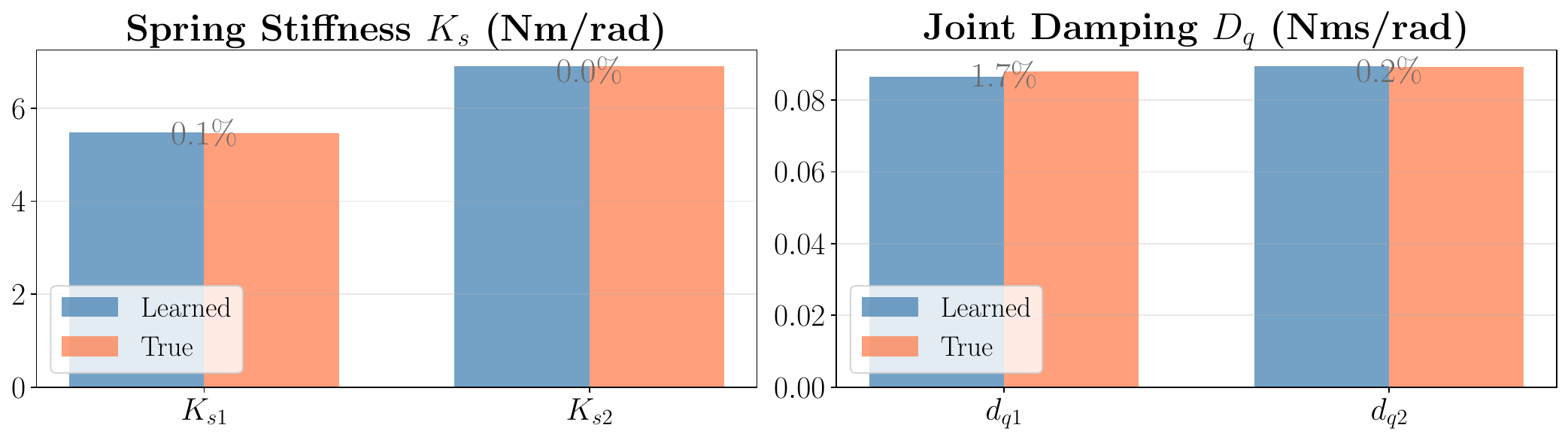}
        \caption{}
        \label{fig:scalar_2link}        
    \end{subfigure} \hspace{-0.7em}
    \begin{subfigure}{0.5\linewidth}
        \centering
        \includegraphics[width=1\linewidth, keepaspectratio]{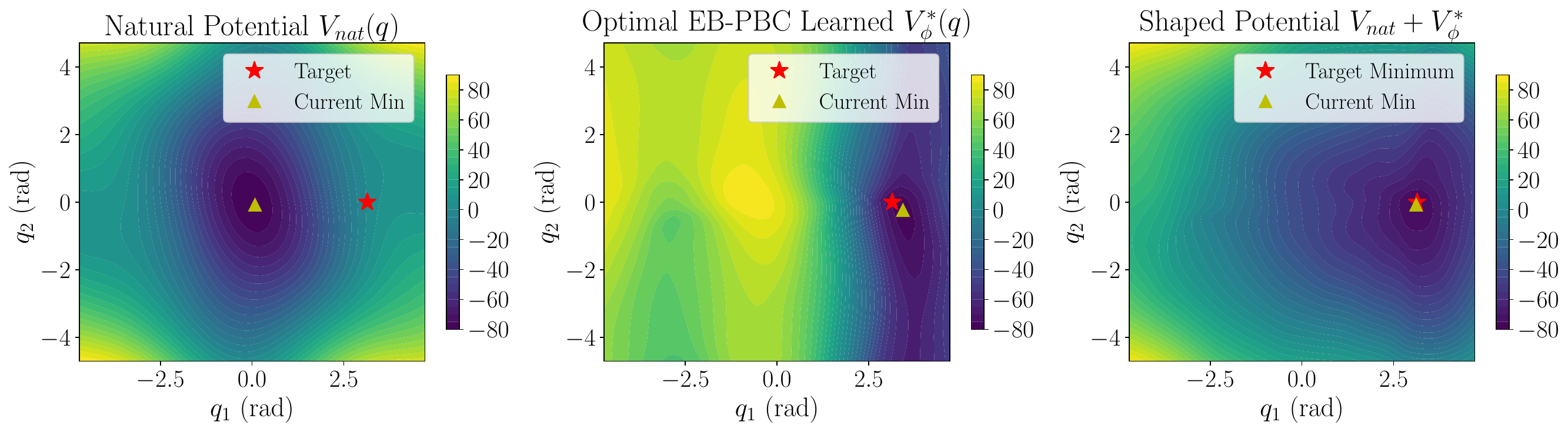}
        \caption{}
        \label{fig:oes_potential_2link}        
    \end{subfigure} \hspace{-3em}
    \\
    \begin{subfigure}{1\linewidth}
        \centering
        \includegraphics[width=1\linewidth, keepaspectratio]{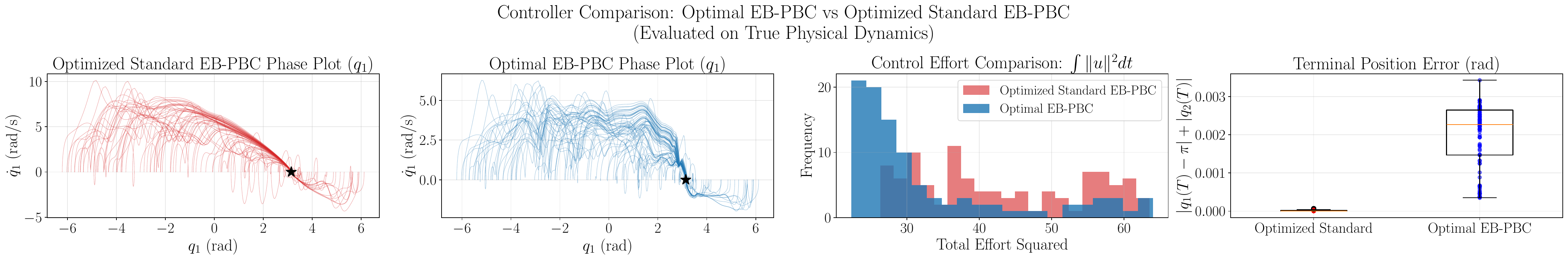}
        \caption{}
        \label{fig:baseline_comparison_2link}      
    \end{subfigure}
    \caption{\textbf{Optimal EB-PBC on $2$-link Torsional Pendulum:} (a) and (b) show the learned pH system with perfectly converged system parameters. Reshaped potential landscape in (c) shows how optimal EB-PBC learns to utilize the natural potential to achieve the minimum at $\bs q^* = [\pi, 0]$. (d) The performance of optimal EB-PBC versus the optimized standard EB-PBC control for swing-up of the torsional pendulum shows how optimally reshaped potential along with time and state-dependent damping leads to massive savings in control effort at the expense of slightly larger terminal position error.}
\end{figure}
Figs.~\ref{fig:M_V_2link},~\ref{fig:scalar_2link} show the mass/inertia matrix, the natural gravitational potential landscape, and the learned system parameters. The constant offset in learned model parameters is resolved by assuming knowledge of the system potential at $\bs q=\bs 0$ and anchoring the learned $V_\theta$. 
%
%

% \subsubsection{Control Performance}
We then compare our proposed optimal EB-PBC control to the standard EB-PBC. We optimized the hyperparameters of the standard EB-PBC for a fair comparison. The optimal parameters for the standard EB-PBC were obtained as $k_p=[28.9, 12.8],\ k_d=[20.42, 0.96]$.
Both controllers successfully stabilize the $2$-link torsional pendulum system at the target configuration $\bs z^* = [\pi, 0, 0, 0]$. However, Figs.~\ref{fig:oes_potential_2link},~\ref{fig:baseline_comparison_2link} show how the proposed controller is able to exploit the system's intrinsic passive dynamics to lower the control effort required to achieve the swing-up.
This performance disparity highlights the fundamental inefficiency inherent in strict potential compensation techniques. The standard EB-PBC wastes significant motor torque by actively canceling the natural torsional elasticity and gravitational pull of the links to enforce an arbitrary quadratic potential. In contrast, the proposed framework bypasses perfect cancellation. It learns an added potential $V_\phi^*(\bs q)$ that reshapes the total energy landscape to place the global minimum at the target, while naturally cooperating with the stored elastic energy of the torsional springs to ``swing" the links toward the target. Notably, mean effort expended for swing-up by our proposed controller is $32.62$ compared to $43.04$ for the standard EB-PBC, achieving a $24.02\%$ control energy savings. The trade-off is a slightly higher terminal position error for our proposed controller.

\section{Conclusions and Future Directions}\label{sec:conclusion}
In this work, we proposed a physics-informed learning framework for optimal energy-balancing passivity-based control of port-Hamiltonian systems. The system model and control policy are jointly learned through alternating optimization with policy-aware data collection. The pH model is updated using a mixture of step-excited and policy-excited trajectory data, concentrating model capacity on the closed-loop operating region, and the EB-PBC policy is re-optimized on the improved model. Both components are parameterized by NNs, the system via structured pH matrices, and the policy via an added potential and state-dependent damping, yielding an interpretable controller expressed as an interplay between the system’s natural and shaped energies. We established deterministic stability guarantees for the learned controller on the true system, with an explicit estimate of the size of the residual neighborhood as a function of model training error. A dissipation regularization mechanism enforces strict energy decay throughout training, strengthening robustness to sim-to-real gaps. The framework was validated on stabilization and swing-up tasks for a $1$-link planar, a $1$-link torsional pendulum, and a $2$-link torsional pendulum. 
Future work will investigate sensitivity to hyperparameters and system variations, and extend the framework to higher-dimensional systems such as flexible-joint manipulators and multi-agent systems. Another possible avenue is to integrate our proposed framework with image-based system identification techniques~\citep{mason2023learning,zhong2020unsupervised}.

\footnotesize 
\bibliographystyle{ieeetr}
\bibliography{pidpg}

@article{massaroli2022optimal,
    author = {Massaroli, Stefano and Poli, Michael and Califano, Federico and Park, Jinkyoo and Yamashita, Atsushi and Asama, Hajime},
    title = {Optimal Energy Shaping via Neural Approximators},
    journal = {SIAM Journal on Applied Dynamical Systems},
    volume = {21},
    number = {3},
    pages = {2126-2147},
    year = {2022},
    doi = {10.1137/21M1414279}
}

@book{schaft1999l2,
  title={$L_2$-Gain and Passivity Techniques in Nonlinear Control},
  author={Van der Schaft, Arjan},
  year={2000},
  publisher={Springer}
}

@inproceedings{
Zhong2020Symplectic,
title={Symplectic {ODE-Net}: Learning {H}amiltonian Dynamics with Control},
author={Yaofeng Desmond Zhong and Biswadip Dey and Amit Chakraborty},
booktitle={International Conference on Learning Representations},
year={2020}
}

@inproceedings{chen2018neural,
 author = {Chen, Ricky T. Q. and Rubanova, Yulia and Bettencourt, Jesse and Duvenaud, David K},
 booktitle = {Advances in Neural Information Processing Systems},
 pages = {},
 title = {Neural Ordinary Differential Equations},
 volume = {31},
 year = {2018}
}

@article{wang2021understanding,
author = {Wang, Sifan and Teng, Yujun and Perdikaris, Paris},
title = {Understanding and Mitigating Gradient Flow Pathologies in Physics-Informed Neural Networks},
journal = {SIAM Journal on Scientific Computing},
volume = {43},
number = {5},
pages = {A3055-A3081},
year = {2021},
doi = {10.1137/20M1318043}
}

@article{bischof2025multi,
  title={Multi-objective loss balancing for physics-informed deep learning},
  author={Bischof, Rafael and Kraus, Michael A},
  journal={Computer Methods in Applied Mechanics and Engineering},
  volume={439},
  pages={117914},
  year={2025},
  publisher={Elsevier}
}

@inproceedings{
loshchilov2017sgdr,
title={{SGDR}: Stochastic Gradient Descent with Warm Restarts},
author={Ilya Loshchilov and Frank Hutter},
booktitle={International Conference on Learning Representations},
year={2017}
}

@article{baydin2018automatic,
  author  = {Atilim Gunes Baydin and Barak A. Pearlmutter and Alexey Andreyevich Radul and Jeffrey Mark Siskind},
  title   = {Automatic Differentiation in Machine Learning: {A} Survey},
  journal = {Journal of Machine Learning Research},
  year    = {2018},
  volume  = {18},
  number  = {153},
  pages   = {1--43},
}

@ARTICLE{sprangers2015reinforcement,
  author={Sprangers, Olivier and Babuška, Robert and Nageshrao, Subramanya P. and Lopes, Gabriel A. D.},
  journal={IEEE Transactions on Cybernetics}, 
  title={Reinforcement Learning for Port-{H}amiltonian Systems}, 
  year={2015},
  volume={45},
  number={5},
  pages={1017-1027},
  doi={10.1109/TCYB.2014.2343194}}

@ARTICLE{ortega2008control,
  author={Ortega, R. and Van Der Schaft, A. and Castanos, F. and Astolfi, A.},
  journal={IEEE Transactions on Automatic Control}, 
  title={Control by Interconnection and Standard Passivity-Based Control of Port-{H}amiltonian Systems}, 
  year={2008},
  volume={53},
  number={11},
  pages={2527-2542},
  doi={10.1109/TAC.2008.2006930}}

@article{gheibi2020designing,
author = {A. Gheibi and A. R. Ghiasi and S. Ghaemi and M. A. Badamchizadeh},
title = {Designing of robust adaptive passivity-based controller based on reinforcement learning for nonlinear port-{H}amiltonian model with disturbance},
journal = {International Journal of Control},
volume = {93},
number = {8},
pages = {1754-1764},
year  = {2020},
publisher = {Taylor & Francis},
doi = {10.1080/00207179.2018.1532607}
}

@article{okura2020bayesian,
author = {Yuki Okura and Kenji Fujimoto and Ichiro Maruta and Akio Saito and Hidetoshi Ikeda},
title = {Bayesian Inference for Path Following Control of Port-{H}amiltonian Systems with Training Trajectory Data},
journal = {SICE Journal of Control, Measurement, and System Integration},
volume = {13},
number = {2},
pages = {40-46},
year  = {2020},
publisher = {Taylor & Francis},
doi = {10.9746/jcmsi.13.40}
}

@article{zhang2015dissipation,
  title={Dissipation Obstacle hampers Control-by-Interconnection Methodology},
  author={Zhang, Meng and Ortega, Romeo and Jeltsema, Dimitri and Su, Hongye},
  journal={IFAC-PapersOnLine},
  volume={48},
  number={13},
  pages={123--128},
  year={2015},
  publisher={Elsevier}
}

@article{sebastian2025physics,
  author={Sebastián, Eduardo and Duong, Thai and Atanasov, Nikolay and Montijano, Eduardo and Sagüés, Carlos},
  journal={IEEE Transactions on Robotics}, 
  title={Physics-Informed Multiagent Reinforcement Learning for Distributed Multirobot Problems}, 
  year={2025},
  volume={41},
  number={},
  pages={4499-4517},
  doi={10.1109/TRO.2025.3582836}}

@inproceedings{plaza2022total,
  title={Total energy shaping with neural interconnection and damping assignment-passivity based control},
  author={Plaza, Santiago Sanchez Escalonilla and Reyes-B{\'a}ez, Rodolfo and Jayawardhana, Bayu},
  booktitle={Learning for Dynamics and Control Conference},
  pages={520--531},
  year={2022},
  organization={PMLR}
}

@article{greydanus2019hamiltonian,
  title={{H}amiltonian neural networks},
  author={Greydanus, Samuel and Dzamba, Misko and Yosinski, Jason},
  journal={Advances in Neural Information Processing Systems},
  volume={32},
  year={2019}
}

@article{duong2024port,
  title={Port-{H}amiltonian neural {ODE} networks on {L}ie groups for robot dynamics learning and control},
  author={Duong, Thai and Altawaitan, Abdullah and Stanley, Jason and Atanasov, Nikolay},
  journal={IEEE Transactions on Robotics},
  volume={40},
  pages={3695--3715},
  year={2024},
  publisher={IEEE}
}

@article{dipersio2025stochastic,
  author  = {Di Persio, Luca and Ehrhardt, Matthias and Outaleb, Youness},
  title   = {Stochastic port-{H}amiltonian neural networks: {U}niversal approximation with passivity guarantees},
  journal = {arXiv preprint arXiv:2603.10078},
  year    = {2025}
}

@inproceedings{van2025learning,
  title={Learning Subsystem Dynamics in Nonlinear Systems via Port-{H}amiltonian Neural Networks},
  author={Van Otterdijk, {GJE} and Moradi, Sarvin and Weiland, Siep and T{\'o}th, Roland and Jaensson, {NO} and Schoukens, Maarten},
  booktitle={IEEE 64th Conference on Decision and Control (CDC)},
  pages={2071--2076},
  year={2025},
  organization={}
}

@article{beckers2023learning,
  title={Learning switching port-{H}amiltonian systems with uncertainty quantification},
  author={Beckers, Thomas and Jiahao, Tom Z and Pappas, George J},
  journal={IFAC-PapersOnLine},
  volume={56},
  number={2},
  pages={525--532},
  year={2023},
  publisher={Elsevier}
}

@inproceedings{beckers2022gaussian,
  title={Gaussian process port-{H}amiltonian systems: Bayesian learning with physics prior},
  author={Beckers, Thomas and Seidman, Jacob and Perdikaris, Paris and Pappas, George J},
  booktitle={IEEE 61st Conference on Decision and Control (CDC)},
  pages={1447--1453},
  year={2022},
  organization={}
}

@article{roth2025stable,
  title={Stable port-{H}amiltonian neural networks},
  author={Roth, Fabian J and Klein, Dominik K and Kannapinn, Maximilian and Peters, Jan and Weeger, Oliver},
  journal={arXiv preprint arXiv:2502.02480},
  year={2025}
}

@article{dai2021lyapunov,
  title={Lyapunov-stable neural-network control},
  author={Dai, Hongkai and Landry, Benoit and Yang, Lujie and Pavone, Marco and Tedrake, Russ},
  journal={arXiv preprint arXiv:2109.14152},
  year={2021}
}

@article{chang2019neural,
  title={Neural {L}yapunov Control},
  author={Chang, Ya-Chien and Roohi, Nima and Gao, Sicun},
  journal={Advances in Neural Information Processing Systems},
  volume={32},
  year={2019}
}

@article{banerjee2025survey,
  title={A survey on physics-informed reinforcement learning: {R}eview and open problems},
  author={Banerjee, Chayan and Nguyen, Kien and Fookes, Clinton and Raissi, Maziar},
  journal={Expert Systems with Applications},
  volume={287},
  pages={128166},
  year={2025},
  publisher={Elsevier}
}

@article{van2014port,
  title={Port-{H}amiltonian systems theory: An introductory overview},
  author={Van Der Schaft, Arjan and Jeltsema, Dimitri},
  journal={Foundations and Trends{\textregistered} in Systems and Control},
  volume={1},
  number={2-3},
  pages={173--378},
  year={2014},
  publisher={Emerald Publishing Limited Boston—Delft}
}

@article{ortega2002putting,
  title={Putting energy back in control},
  author={Ortega, Romeo and Van Der Schaft, Arjan J and Mareels, Iven and Maschke, Bernhard},
  journal={IEEE Control Systems Magazine},
  volume={21},
  number={2},
  pages={18--33},
  year={2001},
}

@inproceedings{xu2025learning,
  title={Learning Neural {K}oopman Operators with Dissipativity Guarantees},
  author={Xu, Yuezhu and Sivaranjani, S and Gupta, Vijay},
  booktitle={IEEE 64th Conference on Decision and Control (CDC)},
  pages={2064--2070},
  year={2025},
  organization={}
}

@article{sivaranjani2025control,
  title={Control-Oriented System Identification: Classical, Learning, and Physics-Informed Approaches},
  author={Sivaranjani, S and Shi, Yuanyuan and Atanasov, Nikolay and Duong, Thai and Feng, Jie and Martin, Tim and Xu, Yuezhu and Gupta, Vijay and Allg{\"o}wer, Frank},
  journal={arXiv preprint arXiv:2512.06315},
  year={2025}
}

@article{mason2023learning,
  title={Learning to predict {3D} rotational dynamics from images of a rigid body with unknown mass distribution},
  author={Mason, Justice J and Allen-Blanchette, Christine and Zolman, Nicholas and Davison, Elizabeth and Leonard, Naomi Ehrich},
  journal={Aerospace},
  volume={10},
  number={11},
  pages={921},
  year={2023},
  publisher={MDPI}
}

@article{zhong2020unsupervised,
  title={Unsupervised learning of {Lagrangian} dynamics from images for prediction and control},
  author={Zhong, Yaofeng Desmond and Leonard, Naomi},
  journal={Advances in Neural Information Processing Systems},
  volume={33},
  pages={10741--10752},
  year={2020}
}

\end{document}